\documentclass[acmsmall]{acmart}
\usepackage{epsfig,subfigure,amsmath,amsfonts,graphicx,bm,color,multirow,stfloats,float,tabu,tabularx,algorithm,algorithmicx,algpseudocode,amsthm, booktabs}

\newtheorem{theorem}{Theorem}[section]
\newtheorem{corollary}{Corollary}[theorem]
\newtheorem{lemma}{Lemma}[theorem]
\newtheorem{remark}{Remark}[theorem]

\AtBeginDocument{%
	}

\setcopyright{acmlicensed}
\acmJournal{TOIT}
\acmYear{2025} \acmVolume{1} \acmNumber{1} \acmArticle{1} \acmMonth{1}\acmDOI{10.1145/3747348}

\acmJournal{TOIT}
\acmVolume{20}
\acmNumber{12}
\acmArticle{333}
\acmMonth{12}

\begin{document}
	
	\title{Integrated User Scheduling and Beam Steering in Over-the-air Federated Learning for Mobile IoT}
	
	\author{Shengheng~Liu}
	\authornote{Corresponding author}
	\orcid{0000-0001-6579-9798}
	\affiliation{%
		\institution{School of Information Science and Engineering, Southeast University}
		\city{Nanjing}
		\postcode{210096}
		\country{China}}
	\affiliation{%
		\institution{Purple Mountain Laboratories}
		\city{Nanjing}
		\postcode{211111}
		\country{China}}
	\email{s.liu@seu.edu.cn}
	
	\author{Ningning~Fu}
	\orcid{0009-0009-7407-6597}
	\affiliation{%
		\institution{School of Information Science and Engineering, Southeast University}
		\city{Nanjing}
		\postcode{210096}
		\country{China}}
	\email{funingning@seu.edu.cn}
	
	\author{Zhonghao~Zhang}
	\affiliation{%
		\institution{China Unicom Research Institute}
		\city{Beijing}
		\postcode{100048}
		\country{China}}
     \email{zhangzh306@chinaunicom.cn}

	\author{Yongming~Huang}
	\orcid{0000-0003-3616-4616}
	\affiliation{%
		\institution{School of Information Science and Engineering, Southeast University}
		\city{Nanjing}
		\postcode{210096}
		\country{China}}
	\affiliation{%
		\institution{Purple Mountain Laboratories}
		\city{Nanjing}
		\postcode{211111}
		\country{China}}
\email{huangym@seu.edu.cn}
	
	\author{Tony~Q.~S.~Quek}
	\orcid{0000-0002-4037-3149}
	\affiliation{%
		\institution{Information System Technology and Design Pillar, Singapore University of Technology and Design}
		\postcode{487372}
		\country{Singapore}}
\email{tonyquek@sutd.edu.sg}
	
	\renewcommand{\shortauthors}{Liu et al.}
	
\begin{abstract}
		The rising popularity of Internet of things (IoT) has spurred technological advancements in mobile internet and interconnected systems. While offering flexible connectivity and intelligent applications across various domains, IoT service providers must gather vast amounts of sensitive data from users, which nonetheless concomitantly raises concerns about privacy breaches. Federated learning (FL) has emerged as a promising decentralized training paradigm to tackle this challenge. This work focuses on enhancing the aggregation efficiency of distributed local models by introducing over-the-air computation into the FL framework. Due to radio resource scarcity in large-scale networks, only a subset of users can participate in each training round. This highlights the need for effective user scheduling and model transmission strategies to optimize communication efficiency and inference accuracy. To address this, we propose an integrated approach to user scheduling and receive beam steering, subject to constraints on the number of selected users and transmit power. Leveraging the difference-of-convex technique, we decompose the primal non-convex optimization problem into two sub-problems, yielding an iterative solution. While effective, the computational load of the iterative method hampers its practical implementation. To overcome this, we further propose a low-complexity user scheduling policy based on characteristic analysis of the wireless channel to directly determine the user subset without iteration. Extensive experiments validate the superiority of the proposed method in terms of aggregation error and learning performance over existing approaches.
\end{abstract}

	
	
	\received{2 February 2024}
	
\maketitle
	
\section{Introduction}
	
	Once deemed unfeasible, the \textbf{Internet of things (IoT)}, a network of interconnected devices that sense, communicate, and process data, has significantly expanded connectivity among individuals and quietly revolutionized various facets of our lives, ranging from smart home, entertainment, healthcare to manufacturing, grids, and more \cite{bian2022machine, can2021privacy}. Fully harnessing its potential requires IoT service providers to train large models using massive quantities of user data for diverse purposes \cite{wang2020preserving}. Nevertheless, this reliance on large volumes of sensitive data poses two major challenges: (i) reluctance among users to share private data due to privacy and security concerns \cite{wei2024demystifying}, and (ii) the likelihood of high latency resulting from the transmission of substantial raw data over limited bandwidth between users and the aggregator \cite{xu2021optimized, Xia2021Wireless}. \textbf{Federated learning (FL)} is usually recommended as a compelling solution \cite{Zhang2024Privacy} to overcome the challenges by enabling collaborative model training while keeping user data local \cite{Chen2024Minimax}. A typical FL protocol generally consists of two iterative stages \cite{Nguyen2023Preserving}: Users download a trainable model from the aggregator, update it with their data, and then transmit the renewed models back to the aggregator for aggregation and further refinement. Despite containing fewer parameters than raw data, the constant iteration and frequent transmissions of a shared model between users and the aggregator still results in considerable communication overhead \cite{Zhao2023Tensor}. Hence, improving the transmission efficiency of FL schemes is essential for optimizing IoT operations \cite{Hamer2020Fedboost}.
	
	Efficient communication within FL systems plays a critical role in reducing the time budget and accelerating the convergence rate of the global model. One straightforward approach to achieve this is by minimizing the amount of transmitted data, which can be achieved through techniques such as gradient sparsification and linear projection \cite{9042352}. Additionally, decreasing the number of communication rounds between users and the aggregator also proves beneficial \cite{Sebastian2019Local}. Different from these generally applicable strategies, specific enhancements are tailored for certain data transmission modes. Conventional multiple-access schemes, such as \textbf{orthogonal frequency-division multiple access (OFDM)} and \textbf{non-orthogonal multiple access (NOMA)}, operate based on the separated-communication-and-computation principle. In an effort to improve the convergence performance, joint bandwidth allocation and user scheduling polices have been designed \cite{fast_convergence, chen2020joint}.
	Similarly, a multi-armed bandit-based user scheduling problem has been formulated, with solutions derived using the \textbf{upper confidence bound (UCB)} algorithm \cite{xia2020multi}. Despite these efforts, traditional schemes sometimes struggle to meet stringent latency requirements, especially in scenarios characterized by high user density and limited bandwidth \cite{9520774}, primarily because the aggregator must decode digital signals before processing.
	In response to these challenges, an innovative wireless aggregation principle of \textbf{over-the-air computation (AirComp)}, has received immediate attention \cite{Sahin2023OTA}.
	By exploiting the waveform superposition property, AirComp enables simultaneous access and dramatically reduces communication latency \cite{Qiao2023Communication}.
	This technique has led to the proposal of a low-latency, multi-access scheme for edge learning \cite{Broadband}.
	Given that most commercially off-the-shelf wireless devices are equipped with digital modulation chips, efforts are directed towards designing customized communication techniques for implementing AirComp without the need for transceiver upgrades \cite{9322334}. This approach is expected to facilitate the integration of AirComp into present infrastructures, which can potentially boost the efficiency of practical FL systems.
	
	Building upon the existing research and motivated by aforementioned challenges, this work investigates the problem of integrated user scheduling and beam steering in over-the-air federated learning for \textbf{distributed intelligence (DI)}. We aim to explore the optimization potential of beam steering vectors while maintaining a reasonable computational complexity. Given the prevalence of low \textbf{signal-to-noise (SNR)} scenarios in mobile IoT \cite{3gpp_tr45820}, we distinguish our work from prior studies by prioritizing the minimization of \textbf{mean square error (MSE)} over the maximization of user count.
		A typical use case involves applying FL for road condition monitoring in vehicular networks within intelligent transportation systems.  This scenario is characterized by blockages and reflections from  buildings and tunnels, as well as highly unstable communication links due to the high speeds of vehicles.
	In such high-noise context, system performance is chiefly constrained by transmission errors.
	The reduction of the MSE between the received and transmitted signals leads to an improved SNR ratio at the receiver, thereby greatly enhancing the reliability.
	On the other hand, our FL system supports reduced transmission power in a given noise level, which contributes to lower energy consumption of terminal devices.
	The underlying combinatorial optimization problem is decoupled into two sub-problems, namely user scheduling and receive beam steering optimization. We utilize the \textbf{difference of convex (DC)} representation to determine the optimal beam steering vector based on a subset of selected users. Conversely, knowledge of the beam steering vector facilitates the optimization of user scheduling.
	To this end, we propose an iterative method that alternates between these two phases until convergence is attained.
	While effective in minimizing MSE, this iterative method is computationally expensive.
	To mitigate this issue, we find that users with channels characterized by smaller maximum angles and larger modulus are preferable for scheduling.
	Drawing upon this insight, we propose a channel-based low-complexity user scheduling policy, which allows for direct determination of the user subset, followed by receive beam steering vector optimization without iteration.
	The technical contributions of this work are threefold.
\begin{itemize}
		\item We formulate an integrated user scheduling and beam steering problem specifically designed for low SNR mobile IoT scenarios in a large-scale AirComp-empowered FL system, which simultaneously achieves improved communication efficacy and learning performance.
		\item We devise a generalized sub-optimal solution, termed the iterative aggregate beam steering and user scheduling method by using the DC technique. This method involves alternately optimizing user scheduling and adjusting beam steering vectors until stability is arrived within the selected user subset.
		\item We further reduce the computational complexity by developing a policy that exploits prior knowledge of users' channel characteristics, which is rigorously proved to yield a superior objective value.
\end{itemize}
	
	The rest of this paper is organized as follows.
	Our system model and problem formulation is provided in section \uppercase\expandafter{\romannumeral2}.
	In Section \uppercase\expandafter{\romannumeral3}, we analyze the joint optimization problem in detail and thus summarize an iterative aggregate beam steering and user scheduling method.
	Subsequently, we streamline the user scheduling algorithm and develop a low-complexity policy in section \uppercase\expandafter{\romannumeral4}.
	Following this, experimental results are presented in Section \uppercase\expandafter{\romannumeral5} to validate our analysis and showcase the exceptional performance of the proposed methods.
	Finally, conclusion is drawn in Section \uppercase\expandafter{\romannumeral6}.
	Main notations and definitions are listed in Table. \ref{tab:notation}.

\begin{table}
		\centering
		\caption{Main Notations and Definitions.}
		\resizebox{1.0\columnwidth}{!}{
		\begin{tabular}{cc|cc}
			\toprule
			Notation & Definition & Notation & Definition \\
			\midrule
			$x$ & A scalar quantity & $\mathcal{K}$, $K$ & The total user set and the total number of the users \\
			$\mathbf{x}$ & A column vector & $N_r$ & The number of antennas in the aggregator \\
			$\mathbf{X}$ &A matrix & $\mathcal{D}_k$ & The local dataset of user $k$ \\
			$\mathbf{x}[d]$ & The $d$-th element of $\mathbf{x}$ & $\mathcal{D}$ & The global dataset \\
			$(\cdot)^\mathsf{H}$& Conjugate transpose & $\mu$ & The learning rate or step size for FL  \\
			$(\cdot)^*$& The optimal value & $i$ & the iteration index of FL\\
			$\mathbb{E}(\cdot)$&Statistical expectation & $\mathbf{w}$ & The model weight parametcers\\
			$\mathcal{O}(\cdot)$& Computational complexity & $\mathbf{s}_k$ & The transmit symbol vector of user $k$ \\
			$\mu$& Mean value & $b_k$ & The transmit coefficient of user $k$ \\
			$\sigma$&Variance value & $\mathbf{n}$ & The additive white Gaussian noise vector\\
			$\mathcal{CN}(\mu, \sigma)$& Complex Gaussian distribution &$\eta$ & The scale factor of the aggregator\\
			$\operatorname{Real}(\cdot)$& Real part extraction & $\mathbf{g}_k$ & The gradients of the loss on local dataset $\mathcal{D}_k$ \\
			$\nabla$& Gradient calculation & $\mathbf{h}_k$ & The channel vector between user $k$ and the aggregator\\
			$\Vert\cdot \Vert$&Modulus of a vector & $\mathbf{m}$ & The receive beam steering vector of the aggregator \\
			$\operatorname{Tr(\cdot)}$& Trace of a matrix&$\mathcal{S}$, $S$ & The selected user set and the number of selected users \\
			\bottomrule
\end{tabular}}
\label{tab:notation}
\end{table}
	
\section{Related Work}
	
In this section, we provide a concise overview of related research on FL and AirComp within the context of DI. This literature review also covers advanced approaches to mitigate the aggregation error.
	
\subsection{Separated-Communication-and-Computation Paradigm for FL}

In the initial stage, FL systems predominantly employed the separated communication and computation paradigm in their aggregation processes.
Several strategies have been developed to improve the transmission efficiency under this paradigm. One common approach focuses on reducing the number of communication rounds and compressing data to lower the transmission load \cite{Sebastian2019Local, 9042352, wu2021fedscr, haddadpour2021federated, shah2023model}. Alternatively, user scheduling has been explored to minimize the number of participants in each round, thereby reducing overall communication costs. For instance, a multi-armed bandit-based user selection strategy was proposed in \cite{xia2020multi} and solved using the UCB algorithm. Similarly, the Eiffel algorithm was designed for resource-constrained environments to improve scheduling efficiency and fairness while minimizing communication throughput \cite{sultana2022eiffel}.

Joint optimization of user scheduling and bandwidth allocation has also shown promise in improving both transmission efficiency and convergence performance \cite{chen2020joint}. Zhou et al. \cite{zhou2022communication} introduced Overlap-FedAvg, an innovative framework that parallels model training and communication by integrating hierarchical computing, data compensation, and Nesterov accelerated gradients. Tailored solutions for specific domains, such as the FedCPF approach for vehicular FL, have further improved communication efficiency through customized local training strategies, partial client participation, and flexible aggregation policies \cite{liu2022fedcpf}. However, these methods requires aggregators to decode signals after reception, which impose significant constraints on the transmission efficiency. To further enhance aggregation efficiency in FL, AirComp is introduced as a novel paradigm for concurrent communication and computation.

\subsection{AirComp Paradigm for FL}

AirComp is a wireless communication technique that exploits the waveform superposition property of wireless channels to conduct computations concurrently during transmission \cite{Cao2024Overview}. Several works have made significant progress in the AirComp-empowered FL framework. For example, a dynamic device scheduling mechanism was proposed, which selects qualified edge devices to transmit their local models with a proper power control policy for model training via AirComp \cite{Du2023Gradient}. Furthermore, Wang et al. \cite{Wang2024Joint} presented OMUAA, an efficient algorithm for adaptive updating of local and global models based on the time-varying communication environment.
	
In the presence of fading channels and noise, AirComp inevitably introduces aggregation errors, which are quantified by the MSE \cite{Zhu2021Over} metric. Such error causes deviations in received parameters from their true values, resulting in a notable drop in prediction accuracy \cite{Yang2020Federated}. This issue is critical for numerous DI systems that demand high performance in terms of accuracy metric. For example, an individual identification application based on motion signature needs to process millions of data points collected through VR technology within a few minutes and achieve a recognition accuracy greater than $95\%$ \cite{Kids}. To address these challenges, Henrik et al. \cite{Hell2023Federated} examined the impact of the aggregation error on the convergence of FL and proposed a retransmission strategy to enhance FL accuracy over resource-constrained wireless networks. Amiri et al. \cite{amiri2020federated} demonstrated that a well-designed transmit power of each user can effectively reduce MSE. Moreover, optimizing the denoising factor alongside transmit power has been shown to yield improved outcomes \cite{9382094}. Xu et al. \cite{Xu2021Learning} generalized a joint policy for scaling factor and receive beamforming vector design by employing a dynamic learning rate. Additionally, a time-varying precoding and scaling scheme was proposed in \cite{9322580}, which facilitates aggregation and progressively mitigates the impact of noise.
	
Accelerating the convergence rate by involving more users is feasible if MSE is maintained at a moderate level \cite{Li2023Dynamic}. Accordingly, recent research has focused on mitigating aggregation errors via user scheduling and beam steering, which shows significant advancements. To name a few, Yang et al. \cite{Yang2020Via} formulated a joint user scheduling and beam steering vector design task with the objective of maximizing the number of users while satisfying MSE constraints. On this basis, improvements were made by utilizing convex constraints instead of the previous non-convex quadratic constraint in deriving the solution \cite{kim2023}. With a focus on decreasing computational complexity, Xu et al. \cite{Xu2024Random} introduced a random aggregate beamforming scheme that generates the beamforming vector through random sampling on the complex sphere. However, this scheme does not fully exploit the optimization potential inherent in the beamforming vector. Nevertheless, the existing literature primarily addresses high-SNR scenarios, seeking to maximize the user count within a fixed MSE constraint to capitalize on user diversity and improve system performance. In the context of low-SNR conditions, excessive aggregation error is the main performance bottleneck, which merits in-depth analysis.

\section{System Model and Problem Formulation}
\subsection{System Model}
	
Consider an FL system that consists of $K$ single-antenna users and an aggregator with $N_r$ antennas. The channel between user $k$ and the aggregator is assumed to be identically complex Gaussian distributed with unit power, i.e., ${\mathbf{h}}_k\sim\mathcal{CN}\left(0,\mathbf{I}\right)$.
	Each user $k$ possesses the local dataset $\mathcal{D}_{k}$, which constitutes the global dataset $\mathcal{D}=\cup_{k\in\mathcal{K}}\mathcal{D}_{k}$.
	Typically, the objective function of an FL task is to find the optimal model ${\mathbf{w}}^{o}$ that minimizes the loss function $P\left(\mathbf{w};\mathcal{D}\right)$, where $\mathbf{w}\in\mathbb{R}^{D}$ is the model weight parameters with the dimension of $D$, $P\left(\mathbf{w};\mathcal{D}\right)$ is the loss value based on $\mathcal{D}$, i.e., $P\left(\mathbf{w}\right)\triangleq\frac{1}{\left|\mathcal{D}\right|}\sum_{n=1}^{\left|\mathcal{D}\right|}Q\left(\mathbf{w};\mathbf{x}_{n},y_{n}\right)$ with $Q$ the loss function, $\left(\mathbf{x}_{n},y_{n}\right)$ the data sample.
	
	Although the centralized gradient decent method can be applied to update the model $\mathbf{w}$ in order to obtain $\mathbf{w}^{o}$, it becomes impractical as the local data set $\mathcal{D}_{k}$ at each user cannot be accessed by others out of the privacy and latency concerns.
	Instead, as shown in Fig.~\ref{fig1}, each user $k$ updates its model based on the local dataset $\mathcal{D}_{k}$, and transmits its model to the aggregator, collaboratively training a shared model.
	
	Then, the global model is updated by averaging the local models from the participating users.
	The local models $\mathbf{w}_{k},k\in\mathcal{K}$ and the global model $\mathbf{w}$ are respectively updated by
	\[
	\mathbf{w}_{k}^{i+1}=\mathbf{w}^{i}-\mu\mathbf{g}_{k}\left(\mathbf{w}^{i}\right),
	\]
	
\begin{equation}
		\mathbf{w}^{i+1}=\frac{1}{K}\sum_{k=1}^{K}\mathbf{w}_{k}^{i+1},\label{eq:update}
\end{equation}
	where superposition $i$ is the iteration index, $\mu$ is the learning rate or step size, $\mathbf{g}_{k}(\mathbf{w}^{i})=\nabla_{\mathbf{w}^{i}}P_{k}\left(\mathbf{w}^{i}\right)$
	with $P_{k}\left(\mathbf{w}^{i}\right)=\frac{1}{\left|\mathcal{D}_{k}\right|}\sum_{n=1}^{\left|\mathcal{D}_{k}\right|}Q\left(\mathbf{w}^{i};\mathbf{x}_{n},y_{n}\right)$
	is the gradients of the loss on local dataset $\mathcal{D}_{k}$ with respect to $\mathbf{w}^{i}$.
	
\begin{figure}[h]
		\centering
\includegraphics[scale=0.66]{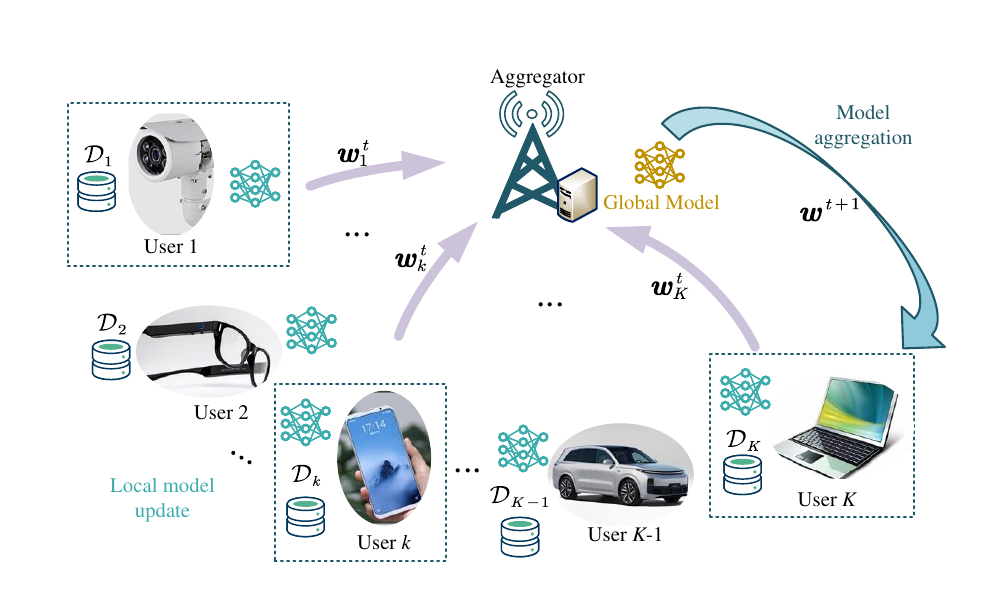}
		\caption{Generic FL architecture for mobile IoT.}
\label{fig1}
\end{figure}
	
	In the considered model, AirComp is used for the model aggregation.
	To be specific, the transmit symbol vector is customized as $\mathbf{s}_{k}^{i+1}\triangleq\mathbf{w}_{k}^{i+1}$,
	which is assumed to be normalized with unit variance, i.e., $\mathbb{E}\left[\mathbf{s}_{k}^{i+1}\left(\mathbf{s}_{k}^{i+1}\right)^{\mathsf{H}}\right]=\mathbf{I}$.
	Then, these symbols are modulated in an analog manner and precoded by the transmit coefficients, after which these signals are superimposed over the air.
	Finally, the received signals at the aggregator are combined by the aggregate beam steering vector $\mathbf{m}$ and the resulted signal is further amplified by a factor $\eta$.
	It is further assumed that the channel is block fading and remained constant during the process of model transmission.
	For notational convenience, the $d$-th element of $\mathbf{s}_{k}$, $\mathbf{w}^{i}$ and $\mathbf{g}_{k}$, i.e., $\mathbf{s}_{k}[d]$, $\mathbf{w}^{i}\left[d\right]$ and $\mathbf{g}_{k}[d]$, are simplified as $s_{k}$, $w^{i}$ and $g_{k}$.
	
	The limited wireless resources deter the participation of all users in the update of the global model, thereby only a subset of users are selected.
	By denoting the selected user set as $\mathcal{S}$, the desired signal based on \eqref{eq:update} is written as:
	
\begin{equation}
		y_{\mathrm{des}}=\sum_{k\in\mathcal{S}}w_{k}^{i+1}=\sum_{k\in\mathcal{S}}s_{k}.
\end{equation}
	Due to the fading and noisy wireless channels, the actual received signal at the aggregator is given by:
	
	\begin{equation}
		y=\sqrt{\eta}\left(\sum_{k\in\mathcal{S}}\mathbf{m}^{\mathsf{H}}\mathbf{h}_{k}b_{k}s_{k}+\mathbf{m}^{\mathsf{H}}\mathbf{n}\right),\label{eq:aggregat_sig}
	\end{equation}
	where $\mathbf{h}_{k}$ is the channel vector from user $k$ to the aggregator, $b_{k}$ is the transmit coefficient at user $k$, $\mathbf{m}$ is the aggregate beam steering vector, $\mathbf{n}$ is the additive white Gaussian noise vector with independent identically distributed $\mathcal{CN}\left(0,\sigma^{2}\right)$.
	
	Thus, the aggregate error via AirComp is expressed as
	
	\begin{equation}
		e\triangleq y_{\mathrm{des}}-y=\sum_{k\in\mathcal{S}}\left(1-\sqrt{\eta}\mathbf{m}^{\mathsf{H}}\mathbf{h}_{k}b_{k}\right)s_{k}-\sqrt{\eta}\mathbf{m}^{\mathsf{H}}\mathbf{n},\label{eq:error}
	\end{equation}
	which is measured by MSE, i.e., $\mathrm{MSE}=\mathbb{E}\left(\left\Vert e\right\Vert ^{2}\right)$,
	including the channel-related and the noise-related components.
	Reducing the aggregate error is critical for the performance of FL task, as a large error would lead to the divergence of the training loss and degrade the classification/regression performance.
	To minimize the error, we can eliminate the error from the fading, while mitigating the error caused by the noise. In this case, the condition and the MSE are respectively written as
	
	\[
	1-\sqrt{\eta}\mathbf{m}^{\mathsf{H}}\mathbf{h}_{k}b_{k}=0,k\in\mathcal{S},
	\]
	\[
	\mathrm{MSE}=\mathbb{E}\left(\left\Vert e\right\Vert ^{2}\right)=\eta\left\Vert \mathbf{m}\right\Vert ^{2}\sigma^{2}.
	\]

\subsection{Problem Formulation}
	
\subsubsection{Minimizing the MSE}

	As discussed in the previous subsection, the MSE needs to be minimized in order to maintain the performance of training and inference.
	The optimization variables include the scaling factor $\eta$, transmitting coefficient \textbf{$b_{k}$}, user subset $\mathcal{S}$, and aggregate beam steering vector $\mathbf{m}$, which are independent from the noise vector $\mathbf{n}$.
	Mathematically, the MSE minimization problem with transmit power constraints can be formulated as
\begin{subequations}	
		\begin{align}
			\min_{\mathbf{m},\eta,b_{k},\mathcal{S\subseteq K}}\quad & \left\Vert \mathbf{m}\right\Vert ^{2}\eta\label{eq:pro1_1-1}\\
			\mathrm{s.t.\quad} & \sqrt{\eta}\mathbf{m}^{\mathsf{H}}\mathbf{h}_{k}b_{k}=1,k\in\mathcal{S}\label{eq:pro1-2-1}\\
			& \left|\mathcal{S}\right|=S\label{eq:pro1_3-1}\\
			& \left\Vert b_{k}\right\Vert ^{2}\leq P,k\in\mathcal{S}, \label{eq:pro1_4-1}
		\end{align}
		\label{eq:SIMO}
\end{subequations}
	where inequality constraint \eqref{eq:pro1_4-1} is the transmit power constraint, constraint \eqref{eq:pro1_3-1} is the fixed number of the selected users. Following \cite{Yang2020Via}, the optimal transmit coefficient $b_{k}$ is designed as $b_{k}=\frac{\mathbf{h}_{k}^{\mathsf{H}}\mathbf{m}}{\sqrt{\eta}\left\Vert \mathbf{m}^{\mathsf{H}}\mathbf{h}_{k}\right\Vert ^{2}}.$
	Substituting $b_{k}$ back to constraint \eqref{eq:pro1_4-1} yields $\eta\geq\frac{1}{P\left\Vert \mathbf{m}^{\mathsf{H}}\mathbf{h}_{k}\right\Vert ^{2}},k\in\mathcal{S}$ for each user $k$, which implies
	
	\begin{equation}
		\eta=\underset{k}{\max}\frac{1}{P\left\Vert \mathbf{m}^{\mathsf{H}}\mathbf{h}_{k}\right\Vert ^{2}},k\in\mathcal{S}.\label{eq:yeta}
	\end{equation}
	Substituting \eqref{eq:yeta} into problem \eqref{eq:SIMO} yields
	
\begin{align}
		\min_{\mathbf{m},\mathcal{S\subseteq K}}\quad & \underset{k\in\mathcal{S}}{\max}\frac{\left\Vert \mathbf{m}\right\Vert ^{2}}{P\left\Vert \mathbf{m}^{\mathsf{H}}\mathbf{h}_{k}\right\Vert ^{2}}\quad\mathrm{s.t.}\;\mathrm{\mathrm{\eqref{eq:pro1_3-1}}}.\label{eq:SIMO1}
\end{align}
	
	For any $\mathbf{m}$, it can be written as the multiplication of its norm and the unit direction vector, i.e., $\mathbf{m}=\left\Vert \mathbf{m}\right\Vert \frac{\mathbf{m}}{\left\Vert \mathbf{m}\right\Vert }.$
	Denote $\tilde{\mathbf{m}}=\frac{\mathbf{m}}{\parallel\mathbf{m}\parallel}$,
	the objective function of problem \eqref{eq:SIMO1} satisfies $\underset{k\in\mathcal{S}}{\max}\frac{\left\Vert \mathbf{m}\right\Vert ^{2}}{P\left\Vert \mathbf{m}^{\mathsf{H}}\mathbf{h}_{k}\right\Vert ^{2}}=\underset{k\in\mathcal{S}}{\max}\frac{\left\Vert \tilde{\mathbf{m}}\right\Vert ^{2}}{P\left\Vert \tilde{\mathbf{m}}^{\mathsf{H}}\mathbf{h}_{k}\right\Vert ^{2}}=\underset{k\in\mathcal{S}}{\max}\frac{1}{P\left\Vert \tilde{\mathbf{m}}^{\mathsf{H}}\mathbf{h}_{k}\right\Vert ^{2}},$
	where $\left\Vert \tilde{\mathbf{m}}\right\Vert =1$. Thereby, problem \eqref{eq:SIMO1} is equivalent to
	
\begin{subequations}
		
		\begin{align}
			\min_{\mathbf{m},\mathcal{S\subseteq K}}\quad & \underset{k\in\mathcal{S}}{\max}\frac{1}{P\left\Vert \mathbf{m}^{\mathsf{H}}\mathbf{h}_{k}\right\Vert ^{2}}\label{eq:SIMO2-1-4}\\
			\mathrm{s.t.}\mathrm{\quad} & \left\Vert \mathbf{m}\right\Vert =1\label{eq:SIMO2-2}\\
			& \mathrm{\mathrm{\eqref{eq:pro1_3-1}}}.\nonumber
		\end{align}
		\label{eq:SIMO1_1}
\end{subequations}
	
	Due to the fact that $\max{\frac{1}{P\left\Vert\mathbf{m}^{\mathsf{H}}\mathbf{h}_{k}\right\Vert^{2}}} = \min{P\left\Vert\mathbf{m}^{\mathsf{H}}\mathbf{h}_{k}\right\Vert^{2}} $, problem \eqref{eq:SIMO1_1} is
	further equivalent to
	
	\begin{subequations}
		
		\begin{align}
			\max_{\mathbf{m},\mathcal{S\subseteq K}}\quad & \underset{k\in\mathcal{S}}{\min}P\left\Vert \mathbf{m}^{\mathsf{H}}\mathbf{h}_{k}\right\Vert ^{2}\label{eq:SIMO2_obj}\\
			\mathrm{s.t.}\mathrm{\quad} &
			 \mathrm{\mathrm{\eqref{eq:pro1_3-1}}},\eqref{eq:SIMO2-2}.\nonumber
		\end{align}
		\label{eq:SIMO2}
	\end{subequations}
	
	Obviously, the above problem \eqref{eq:SIMO2} is an NP-hard combinatorial optimization problems. Its hardness is proved in Appendix \ref{hardness}. A straight-forward method would be to search the user subset in a brute-force manner. However, the size of the search space, i.e., $\tbinom{K}{S}$, becomes prohibitively large when $K$ is large, which makes it intractable. Thus, the user scheduling optimization problem aiming at the minimization of MSE calls for a new method.

\section{Iterative User Scheduling and Aggregate Beam Steering}
		
In this section, we propose an iterative sub-optimal method to solve the combinatorial optimization problem, by alternately solving user scheduling solution $\mathcal{S}$ and the aggregate beamforming vector $\mathbf{m}$.
	Before presenting the solution-obtaining method, we first give the following lemma about the optimal solution.
	\begin{lemma}
		\label{thm:infinity} If $\mathbf{m}^{*}$ is the optimal solution
		of problem \eqref{eq:SIMO2}, $\mathbf{m}^{*}{\mathrm e}^{{\mathrm j}\theta},\forall\theta\in\mathbb{R}$
		is also the optimal solution, which indicates that there are infinite optimal solutions for $\mathbf{m}$.
	\end{lemma}
	\begin{proof}
		For any $\theta\in\mathbb{R}$, we have $\left\Vert \left(\mathbf{m}^{*}{\mathrm e}^{{\mathrm j}\theta}\right)^{\mathsf{H}}\mathbf{h}_{k}\right\Vert ^{2}=\left\Vert \left(\mathbf{m}^{*}\right)^{\mathsf{H}}\mathbf{h}_{k}\right\Vert ^{2}$
		and $\left\Vert \mathbf{m}^{*}{\mathrm e}^{{\mathrm j}\theta}\right\Vert =\left\Vert \mathbf{m}^{*}\right\Vert $,
		which indicates that solutions $\mathbf{m}^{*}{\mathrm e}^{{\mathrm j}\theta}$ and
		$\mathbf{m}^{*}$ have the same objective value and guarantee constraint \eqref{eq:SIMO2-2}. Thus, $\mathbf{m}^{*}{\mathrm e}^{{\mathrm j}\theta},\forall\theta\in\mathbb{R}$
		is also the optimal solution of problem \eqref{eq:SIMO2}, if $\mathbf{m}^{*}$
		is the optimal. This completes the proof.
	\end{proof}
	\begin{remark}
The physical beam direction remains unchanged under a uniform global phase rotation, since the beam formed by the antenna array depends on the relative phases of $\mathbf{m}$, rather than the absolute overall phase.
		At the receiver, the received signal undergoes a global phase rotation by a factor of ${\mathrm e}^{{\mathrm j}\theta}$, which does not alter the signal’s energy distribution.
		This implies that the beamforming objective is invariant to global phase rotations, and the optimal solution is therefore not unique.
		In practical, this non-uniqueness can be broken by introducing a phase normalization constraint without loss of optimality. For example, the real part of $\mathbf{m}$ is specified positive, which is convenient for device consistency.
	\end{remark}
	
\subsection{Aggregate Beam Steering Optimization}
	
	Given the selected users subset $\mathcal{S}$, we readily arrive at the sub-problem
	expressed as follows
	\begin{align}
		\max_{\mathbf{m}}\quad & \underset{k\in S}{\min}P\left\Vert \mathbf{m}^{\mathsf{H}}\mathbf{h}_{k}\right\Vert ^{2}\quad\mathrm{s.t.}\;\mathrm{\eqref{eq:SIMO2-2}},\label{eq:SIMO3}
	\end{align}
	which is non-convex due to constraint \eqref{eq:SIMO2-2}. Defining
	a semidefinite matrix $\mathbf{M=}\mathbf{m}\mathbf{m}^{\mathsf{H}}$,
	sub-problem \eqref{eq:SIMO3} is converted into
	\begin{subequations}
		\begin{align}
			\max_{\mathbf{M}}\quad & \underset{k\in S}{\min}P\cdot\mathrm{Tr}\left(\mathbf{h}_{k}\mathbf{h}_{k}^{\mathsf{H}}\mathbf{M}\right)\label{eq:SIMO2-1-2-1}\\
			\mathrm{s.t.}\mathrm{\quad} & \mathrm{Tr}\left(\mathbf{M}\right)=1\label{eq:SIMO2-2-2-1}\\
			& \mathrm{Rank}\left(\mathbf{M}\right)=1.\label{eq:rankone}
		\end{align}
	\end{subequations}
	
The reformulation introduces a semidefinite matrix variable $\mathbf{M}$, under which the originally non-linear objective of $\mathbf{m}$ becomes linear in $\mathbf{M}$, and the original non-convex constraints \eqref{eq:SIMO2-2} also becomes convex feasible set excluding the rank constraint \eqref{eq:rankone}.
	Thus, the original problem is transformed into a \textbf{semidefinite programming (SDP)} problem with a rank-1 constraint.
	The most commonly used semidefinite relaxation (SDR) method first drops the rank-1 constraint. Although the relaxed problem yields an objective value close to that of the original problem \cite{luo2007approximation}, the resulting $\mathbf{M}$ tends to be high-rank. This indicates simultaneous beamforming along multiple orthogonal directions, which is infeasible in practical systems that support only a single beam. Several methods have been developed to address the rank-one constraint, including nuclear norm relaxation, penalty-based techniques, and difference-of-convex (DC) programming. Among them, DC programming is particularly attractive as it provides a tighter approximation to the rank function than nuclear norm and offers better convergence guarantees than simple penalty methods \cite{lu2016nonconvex}. By reformulating the rank-one constraint as a difference between convex functions, DC-based algorithms can iteratively approach rank-one solutions while solving convex subproblems at each step.

By adopting DC representation, \eqref{eq:rankone} is rewritten as $\mathrm{Tr}\left(\mathbf{M}\right)-\left\Vert \mathbf{M}\right\Vert _{2}=0$.
	By introducing an auxiliary variable $\tau$, the problem can be further
	converted to 
\begin{subequations}		
\begin{align}
			\max_{\mathbf{M},\tau}\quad & \tau\label{eq:SIMO2-1-2-1-1}\\
			\mathrm{s.t.}\mathrm{\quad} & \mathrm{Tr}\left(\mathbf{h}_{k}\mathbf{h}_{k}^{\mathsf{H}}\mathbf{M}\right)\geq\tau,\forall k\label{eq:SIMO2-proj}\\
			& \mathrm{Tr}\left(\mathbf{M}\right)-\left\Vert \mathbf{M}\right\Vert _{2}=0\nonumber \\
			& \eqref{eq:SIMO2-2-2-1}.\label{eq:rankone-1}
		\end{align}
		\label{eq:SIMO2_equ}
\end{subequations}
	
	It can be readily proved that $\mathrm{Tr}\left(\mathbf{h}_{k}\mathbf{h}_{k}^{\mathsf{H}}\mathbf{M}\right)\in\left[0,\mathbf{h}_{k}^{\mathsf{H}}\mathbf{h}_{k}\right]$.
	To guarantee constraint \eqref{eq:SIMO2-proj}, $\tau$ should be
	within the range of $\left[\tau^{\mathrm{low}},\tau^{\mathrm{up}}\right]$
	with $\tau^{\mathrm{low}}=0$ and $\tau^{\mathrm{up}}=\min_{k\in\mathcal{S}}\left(\mathbf{h}_{k}^{\mathsf{H}}\mathbf{h}_{k}\right)$.
	Since the objective function includes only one variable, this motivates to use the bisection method. With the bisection method, the objective function is fixed at each step and only the feasibility problem remains to be solved.
	Given $\tau$, the above problem is a check problem.
	It is feasible if the solution of the following problem is rank one, otherwise infeasible.
	
\begin{align}
		\min_{\mathbf{M}} & \mathrm{Tr}\left(\mathbf{M}\right)-\left\Vert \mathbf{M}\right\Vert _{2}\mathrm{\quad}\mathrm{s.t.}\mathrm{\;}\mathrm{\eqref{eq:SIMO2-proj}},\eqref{eq:SIMO2-2-2-1}.\label{eq:SIMO2_equ}
\end{align}
	
	Convex function $\left\Vert \mathbf{M}\right\Vert _{2}$ is approximated by $\mathrm{Real}\left(\mathrm{Tr}\left(\mathbf{v}\mathbf{v}^{\mathsf{H}}\mathbf{M}\right)\right)$
	with \textbf{$\mathbf{v}$ }the eigenvector of the largest eigenvalue of $\mathbf{M}^{t-1}$ where the superscript denotes the $(t-1)$-th iteration.
	Thereby, problem (\ref{eq:SIMO2_equ}) can be solved by iteratively solving
	\begin{align}
		\min_{\mathbf{M}} & \mathrm{Tr}\left(\mathbf{M}\right)-\mathrm{Real}\left(\mathrm{Tr}\left(\mathbf{v}\mathbf{v}^{\mathsf{H}}\mathbf{M}\right)\right)\mathrm{\quad}\mathrm{s.t.}\mathrm{\;}\mathrm{\eqref{eq:SIMO2-proj}},\eqref{eq:SIMO2-2-2-1}.\label{eq:SIMO2_equ_iter}
	\end{align}
	The bisection based aggregate beam steering vector design algorithm is summarized in Algorithm \ref{SubA}.
\begin{algorithm}[t] 	
	\caption{Aggregate Beam Steering Design.}
	\begin{algorithmic}[1]
		\Statex \textbf{Input} $\mathbf{h}_{k}$, $\mathcal{S}$
		\Statex \textbf{Output}: $\mathbf{m}$
		\Statex \textbf{Initialize}: $\mathbf{M}^0$, $\tau^{\mathrm{low}}=0$, $\tau^{\mathrm{up}}=\min\left(\mathbf{h}_{k}^{\mathsf{H}}\mathbf{h}_{k}\right)$
		\State \textbf{while} $\tau^{\mathrm{up}}-\tau^{\mathrm{low}}>\delta$
		\State \quad $\tau=(\tau^{\mathrm{up}}+\tau^{\mathrm{low}})/2$
		\State \quad solve problem \eqref{eq:SIMO2_equ_iter} iteratively until convergence
		\State \quad \textbf{if} $\mathrm{Rank(\mathbf{M})}>1$
		\State \qquad problem \eqref{eq:SIMO2_equ} is infeasible under $\tau$
		\State \qquad $\tau^{\mathrm{up}}=\tau$
		\State \quad \textbf{else if} $\mathrm{Rank(\mathbf{M})}=1$
		\State \qquad problem \eqref{eq:SIMO2_equ} is feasible under $\tau$
		\State \qquad $\tau^{\mathrm{low}}=\tau$
	\end{algorithmic}
	\label{SubA}
\end{algorithm}

\begin{algorithm}[t] 	
	\caption{Iterative Aggregate Beam Steering and User Scheduling.}
	\begin{algorithmic}[1]
		\Statex \textbf{Input}: $\mathbf{h}_{k}$, $S$
		\Statex \textbf{Output}: $\mathbf{m}$, $\mathcal{S}$
		\State Randomly select $S$ users from user set $\mathcal{K}$
		\State \textbf{do} loop
		\State \quad Given the selected users set $\mathcal{S}$, obtain $\mathbf{m}$     using  \textbf{Algorithm} \ref{SubA}
		\State \quad Given the aggregated beam steering vector $\mathbf{m}$, calculate the projection $\left\Vert \mathbf{m}^{\mathsf{H}}\mathbf{h}_{k}\right\Vert$
		\State \quad Sort the projection $\left\Vert \mathbf{m}^{\mathsf{H}}\mathbf{h}_{k}\right\Vert$ in a descending order
		\State \quad Select the user with the largest $S$ projection values, and thus obtaining the selected user set $\mathcal{S}$
		\State \textbf{until} $\mathcal{S}$ converges.
	\end{algorithmic}
	\label{Iterative}
\end{algorithm}
Theorem \ref{the:bound} reveals the upper bound on the gap between the solution obtained by the DC-based method and the optimal solution, and characterizes its convergence rate. The proof can be found in Appendix \ref{proof_bound}.
\begin{theorem}
	\label{the:bound}
Let $f(\mathbf{m})=\Vert\mathbf{m}^\mathsf{H} \mathbf{h}_k \Vert^2$ be the quadratic function defined on $\mathbb{C}^{N_r}$. Suppose $f$ satisfies Polyak-\L ojasiewicz (PL) inequality and $\bar{\mathbf{m}}$ is a critical point, that is, $\nabla f(\bar{\mathbf{m}}) = 0$, there exists a constant $\eta>0$ such that $\vert f(\mathbf{m}) - f(\bar{\mathbf{m}})\vert \leq \frac{1}{2\eta}\Vert \mathbf{H}_k\mathbf{m} \Vert^2$, where $\mathbf{H}_k=\mathbf{h}_k\mathbf{h}_k^\mathsf{H}$.
Moreover, we can derive that the DC algorithm converges at a linear rate.
\end{theorem}

\subsection{User Scheduling}
	
Once the aggregate beam steering vector $\mathbf{m}$ is obtained,
	problem \eqref{eq:SIMO2} is reduced into the determination of selected
	user subset $\mathcal{S}$, formulated as
	\begin{align}
		\max_{\mathcal{S\subseteq K}} & \underset{k}{\min}P\left\Vert \mathbf{m}^{\mathsf{H}}\mathbf{h}_{k}\right\Vert ^{2}\quad\mathrm{s.t.}\mathrm{\;}\mathrm{\mathrm{\eqref{eq:pro1_3-1}}},\label{eq:SIMO2_subB}
	\end{align}
	The sub-problem can be interpreted to find $S$ users with the largest
	projection values $\left\Vert \mathbf{m}^{\mathsf{H}}\mathbf{h}_{k}\right\Vert $.
	Its optimal solution is obtained accordingly.
	
	The overall solving procedure for problem \eqref{eq:SIMO} involves
	alternately solving sub-problem \eqref{eq:SIMO3} and \eqref{eq:SIMO2_subB}
	until convergence. In each iteration, solving sub-problem \eqref{eq:SIMO3}
	involves checking problem \eqref{eq:SIMO2_equ} many times, the computational
	complexity of which is $\mathcal{O}\left(N_r^{6}\log_{2}\frac{1}{\delta}\right)$.
	Obtaining the optimal solution of sub-problem \eqref{eq:SIMO2_subB}
	is much easier with $\mathcal{O}\left(K\log(K)\right)$. Algorithm \ref{Iterative}
	gives the summary of the proposed iterative method. Note that the
	iteration is ended when the selective subset $\mathcal{S}$ is converged
	rather than the aggregate beam steering vector $\mathbf{m}$, which
	is due to the numerous solutions for the same objective according
	to Lemma \ref{thm:infinity}.	
	
\section{Channel-based User Scheduling Policy}
	
	The previous section gives the sub-optimal method by iteratively
	solving the aggregate beam steering $\mathbf{m}$ and the user
	subset $\mathcal{S}$. However, the computational complexity of this
	method can be intense as the solution-obtaining requires many iterations.
	To tackle this challenge, a channel-based user scheduling policy
	is proposed in this section to determine the selected user subset
	$\mathcal{S}$ first, so that sub-problem \eqref{eq:SIMO2_equ} is
	solved only once to obtain the aggregate beam steering vector $\mathbf{m}$.
	
	We first present the analysis on the inner product of the channel
	vectors, which offer the fundamentals for the proposed low-complexity
	user scheduling policy. For simplicity, we first consider the channel
	vectors $\mathbf{h}_{i}$ with the same modulus, which can be
	assumed such that $\left\Vert \mathbf{h}_{i}\right\Vert =1,\forall i\in\mathcal{S}$.
	The angle between any two channel vectors of users $i$ and $j$
	is defined as
	\begin{equation}
		\alpha_{i,j}=\arccos\frac{\left\Vert \mathbf{h}_{i}^{\mathsf{H}}\mathbf{h}_{j}\right\Vert }{\left\Vert \mathbf{h}_{i}\right\Vert \left\Vert \mathbf{h}_{j}\right\Vert }=\arccos\left\Vert \mathbf{h}_{i}^{\mathsf{H}}\mathbf{h}_{j}\right\Vert \in\left[0,\frac{\pi}{2}\right],
	\end{equation}
	which also equals to $\alpha_{j,i}$. Similarly, the angle between
	the optimal aggregate beam steering vector $\mathbf{m}^{*}$ and
	$\mathbf{h}_{i},\forall i\in\mathcal{S}$ is
	\begin{equation}
		\theta_{i}=\arccos\frac{\left\Vert \mathbf{m}^{*\mathsf{H}}\mathbf{h}_{i}\right\Vert }{\left\Vert \mathbf{m}^{*}\right\Vert \left\Vert \mathbf{h}_{i}\right\Vert }=\arccos\left\Vert \mathbf{m}^{*\mathsf{H}}\mathbf{h}_{i}\right\Vert \in\left[0,\frac{\pi}{2}\right].
	\end{equation}
	
	Let $\alpha=\max_{\forall i,j\in\mathcal{S}}\alpha_{i,j}$ be the maximum angle between any two channel vectors, we have the following theorem in terms of $\alpha$ and $\theta_{i}$.
	\begin{theorem}
		\label{thmtheorem1}Given the selected user subset $\mathcal{S}$
		and further assume $\left\Vert \mathbf{h}_{i}\right\Vert =1,\forall i\in\mathcal{S}$,
		we have $\alpha\geq\theta_{i}\geq0,\forall i\in\mathcal{S}$.
	\end{theorem}
	\begin{proof}
		The theorem can be proved by contradiction. Assume that there exists a user $k\in\mathcal{S}$ such that $\alpha<\theta_{k}$, then we have
		\begin{equation}
			\left\Vert \mathbf{m}^{*\mathsf{H}}\mathbf{h}_{k}\right\Vert =\cos\theta_{k}<\cos\alpha.\label{eq:lemma1}
		\end{equation}
		Without loss of generality, we can assume that the angle between $\mathbf{h}_{p}$
		and $\mathbf{h}_{q}$ has the maximum value, i.e., $\left\Vert \mathbf{h}_{p}^{\mathsf{H}}\mathbf{h}_{q}\right\Vert =\cos\alpha$.
		By designing the aggregate beam steering vector $\mathbf{m}$ as
		$\mathbf{h}_{p}$ (or $\mathbf{h}_{q}$), we can derive that
		\begin{equation}
			\min_{i\in\mathcal{S}}\left(\left\Vert \mathbf{m}^{\mathsf{H}}\mathbf{h}_{i}\right\Vert \right)=\cos\alpha\leq\left\Vert \mathbf{m}^{*\mathsf{H}}\mathbf{h}_{k}\right\Vert =\cos\theta_{k},\label{eq:lemma2}
		\end{equation}
		which is contradict to (\ref{eq:lemma1}). Thus, the presumption that
		there exists a user $k\in\mathcal{S}$ with $\theta_{k}>\alpha$
		does not hold, and the proof is completed.
	\end{proof}
	\begin{corollary}
		Given the selected user subset $\mathcal{S}$ and further assume
		$\left\Vert \mathbf{h}_{i}\right\Vert =1,\forall i\in\mathcal{S}$,
		the optimal objective value\textup{ $\mathrm{obj^{*}}$} is lower
		bounded by $P {\cos}^2\alpha$, i.e., $\mathrm{obj^{*}}\geq P{\cos}^2\alpha$.
	\end{corollary}
	\begin{remark}
		It can be summarized from Theorem 1 and Corollary 1 that users whose channels have a smaller angle are probability to decrease the maximum angle $\alpha$ of the subset $\mathcal{S}$ and can increase the lower bound of the objective value\textup{ $\mathrm{obj^{*}}$}. In wireless communications, phase deviation can be intuitively interpreted as transmission delays resulting from differences in path length, refraction, diffraction, and other factors. Users with smaller phase angles exhibit better alignment, allowing for stronger constructive superposition in the intended direction and more precise signal delivery. The improvement of the proposed selection criterion over random user selection can be found in Fig. \ref{MSE-K_Nr}.
		Therefore, these users are in preference to others when selecting the subset $\mathcal{S}$ .
	\end{remark}
	
	Without loss of generality, we assume that $\left\Vert \mathbf{m}^{*\mathsf{H}}\mathbf{h}_{i}\right\Vert \geq\left\Vert \mathbf{m}^{*\mathsf{H}}\mathbf{h}_{j}\right\Vert ,\forall i\leq j$
	with $\mathbf{m}^{*}$ the optimal solution, then the optimal
	objective value of problem \eqref{eq:SIMO2} is $\mathrm{obj^{*}}=P\left\Vert \mathbf{m}^{*\mathsf{H}}\mathbf{h}_{S}\right\Vert ^{2}$.
	Considering the general case where the channels have various gains such that $\bar{\mathbf{h}}_{s}=\left\Vert \bar{\mathbf{h}}_{s}\right\Vert /\mathbf{h}_{s}$
	with $\left\Vert \bar{\mathbf{h}}_{s}\right\Vert \geq1$, we have the following theorem.
	
	\begin{theorem}
		\label{thmtheorem2}
		Denoting the optimal objective value under $\mathbf{\bar{h}}_{s}$
		as $\mathrm{\overline{obj}}^{*}$, we have $\mathrm{\overline{obj}}^{*}\geq\mathrm{obj^{*}}$
		if $\left\Vert \bar{\mathbf{h}}_{s}\right\Vert \geq\left\Vert \mathbf{h}_{s}\right\Vert ,\forall s\in\mathcal{S}$,
		where $\bar{\mathbf{h}}_{s}/\left\Vert \bar{\mathbf{h}}_{s}\right\Vert =\mathbf{h}_{s}/\left\Vert \mathbf{h}_{s}\right\Vert $.
	\end{theorem}
	\begin{proof}
		Since channel vectors $\mathbf{\bar{h}}_{s}$ and $\mathbf{h}_{s}$
		have the same direction, it is obvious that $\left\Vert \mathbf{m}^{\mathsf{H}}\bar{\mathbf{h}}_{s}\right\Vert \geq\left\Vert \mathbf{m}^{\mathsf{H}}\mathbf{h}_{s}\right\Vert ,\forall s\in\mathcal{S}$
		for any given aggregate beam steering vector $\mathbf{m}$. As a
		result, we have
		\begin{equation}
			\left\Vert \mathbf{m}^{*\mathsf{H}}\bar{\mathbf{h}}_{s}\right\Vert \geq\left\Vert \mathbf{m}^{*\mathsf{H}}\mathbf{h}_{s}\right\Vert \geq\left\Vert \mathbf{m}^{*\mathsf{H}}\mathbf{h}_{S}\right\Vert ,\forall s\in\mathcal{S},
		\end{equation}
		which implies $\underset{s}{\min}P\left\Vert \mathbf{m}^{*\mathsf{H}}\bar{\mathbf{h}}_{s}\right\Vert ^{2}\geq P\left\Vert \mathbf{m}^{*\mathsf{H}}\mathbf{h}_{S}\right\Vert ^{2}=\mathrm{obj^{*}}.$
		Furthermore, as $\mathbf{m}^{*}$ may not be the optimal solution
		under $\mathbf{\bar{h}}_{s}$, the following inequality holds,
		i.e.,
		\[
		\mathrm{\overline{obj}}^{*}\geq\underset{s\in\mathcal{S}}{\min}P\left\Vert \mathbf{m}^{*\mathsf{H}}\bar{\mathbf{h}}_{s}\right\Vert ^{2}\geq\mathrm{obj^{*}}.
		\]
		The proof is complete.
	\end{proof}
	\begin{remark}
		Theorem 2 provides designers with another inspiration.
		Those users with larger channel vector modulus tend to bring higher objective value, so they should be given priority to when selecting the subset $\mathcal{S}$. This can be intuitively understood as follows: channels with larger moduli are less susceptible to path loss, shadow fading, and small-scale fading. In other words, they exhibit stronger channel gains, resulting in higher received power.
	\end{remark}
	
	From the above analysis on the angles between the channels as well as their modulus, we can gain the following observation. That is, the users whose channels have a smaller maximum angle $\alpha$ and lager modulus $\left\Vert \mathbf{h}\right\Vert $ are more likely to achieve a better objective value. Based on the observation, we consider the absolute of inner product $\left\Vert \mathbf{h}_{i}^{\mathsf{H}}\mathbf{h}_{j}\right\Vert ,\forall i\leq j\in\mathcal{K}$ as a condition to select the users.
	
	Firstly, a user should be selected to initialize the subset $\mathcal{S}$. One approach named as ``policy'' is to select the user with the largest channel vector modulus $\left\Vert{\mathbf{h}}\right\Vert$ directly, whose process is illustrated in Fig. \ref{figure_policy} in detail. That is, $\mathcal{S} \leftarrow \{ \arg\max_{k\in\mathcal{K}}\left\Vert\mathbf{h}_k\right\Vert \}$. In each subsequent step, a new user is selected based on the inner product, and then added to the user subset $\mathcal{S}$. Specifically, in order to select the $s$-th user, we calculate the maximum inner product $ \left( \min_{s\in\mathcal{S'}}\left\Vert \mathbf{h}_{s}^{\mathsf{H}}\mathbf{h}_{i}\right\Vert \right) $ for each unselected user $i \notin \mathcal{S}$. Then we select the user $j$ which satisfies the following criteria
	\begin{equation}
		j = \arg\max_{i\in\mathcal{K}\backslash\mathcal{S'}}\left(\min_{s\in\mathcal{S'}}\left\Vert \mathbf{h}_{s}^{\mathsf{H}}\mathbf{h}_{i}\right\Vert \right).
	\end{equation}
	Then user $j$ can be included in subset $\mathcal{S}$ as follows
	\begin{equation}
		\mathcal{S'}\leftarrow\mathcal{S'}\cup\left\{ j\right\}.
	\end{equation}
	
	\begin{figure}
		\centering\includegraphics[scale=0.35]{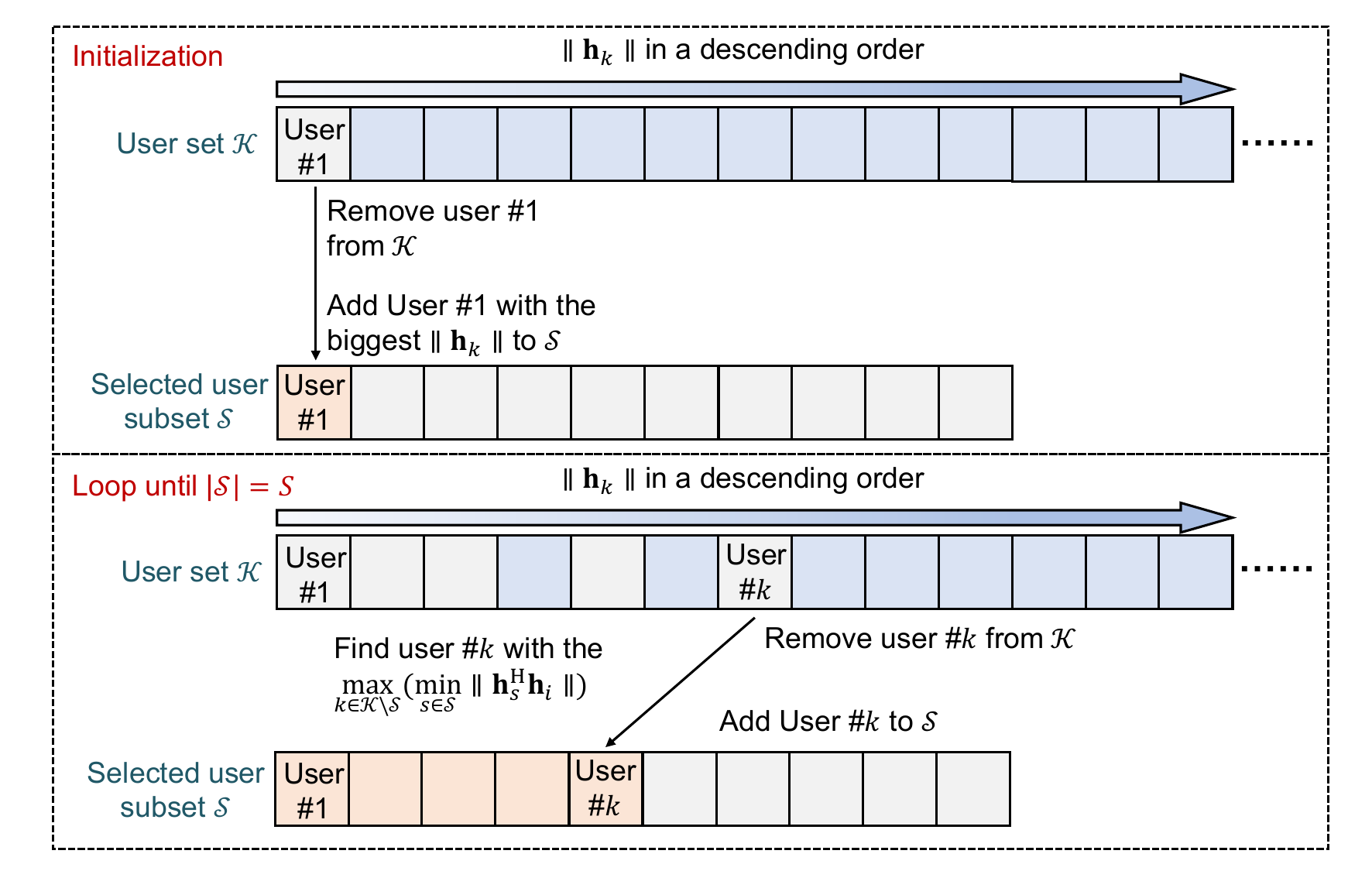}
		\caption{An illustration of low-complexity user scheduling policy.}
		\label{figure_policy}
	\end{figure}
	
	Nonetheless, the marginal difference in the channel vector modulus implies that the channels with the second or third largest modulus may possess a reduced maximum angle $\alpha$ with other users. For this reason, the scheduling criterion for the leading user within the ``Policy\_greedy'' diverges from the conventional ``Policy'' that prioritizes the user with the largest channel vector modulus.
	An illustration of the ``Policy\_greedy'' is presented in Fig. \ref{figure_policy_greedy}.
	Instead, ``Policy\_greedy'' systematically identifies the initial user from a subset comprised of $G$ users that exhibit the topmost $G$ channel gains, where $G$ is predetermined to be $G=5$.
	Consequently, the selected users subset $\mathcal{S}$ is obtained with the largest $\mathrm{obj}$ value among the $G$ considered objective values.
	The comprehensive procedure of ``Policy\_greedy'' is encapsulated in Algorithm \ref{Scheduling}, which degenerates to the ``Policy'' method when $G=1$.
	It can be expected that the ``Policy\_greedy'' will outperform the standard ``Policy'', which can be verified in the subsequent simulation outcomes.
	\begin{algorithm}[t] 	
		\caption{Greedy Low-complexity User Scheduling Policy.}
		\begin{algorithmic}[1]
			\Statex \textbf{Input}: $\mathbf{h}_{k}$, $G$
			\Statex \textbf{Output}: $\mathcal{S}$
			\Statex \textbf{Initialize}: $\mathbf{val}=0$, $c = 0$
			\State sort the channel vector modulus $\left\Vert \mathbf{h}_{k}\right\Vert$ in a descending order
			\State \textbf{for} $k=1:K$
			\State \quad $\mathcal{S'}\leftarrow \left\{k\right\}$
			\State \quad \textbf{for} $s=1:S-1$
			
			\State \qquad $\mathrm{tmp}=\max_{i\in\mathcal{K}\backslash\mathcal{S'}}\left(\min_{s\in\mathcal{S'}}\left\Vert \mathbf{h}_{s}^{\mathsf{H}}\mathbf{h}_{i}\right\Vert \right)$
			
			\State \qquad \textbf{if} $\mathrm{tmp}<\min{\mathbf{val}}$
			\State \qquad\quad break
			
			\State \qquad $j = \arg\max_{i\in\mathcal{K}\backslash\mathcal{S'}}\left(\min_{s\in\mathcal{S'}}\left\Vert \mathbf{h}_{s}^{\mathsf{H}}\mathbf{h}_{i}\right\Vert \right)$
			
			\State \qquad $\mathcal{S'}\leftarrow\mathcal{S'}\cup\left\{ j\right\}$
			\State \quad \textbf{end for}
			\State \quad \textbf{if} $c < G$
			\State \qquad $\mathbf{val}_c = \mathrm{tmp}$
			\State \qquad $\mathcal{S}_{c} = \mathcal{S'}$
			\State \qquad $ c = c +1$
			\State \quad \textbf{else if} $\mathrm{tmp} > \min{\mathbf{val}} $
			\State \qquad $\mathrm{index} = \arg\min{\mathbf{val}}$
			\State \qquad $\mathcal{S}_{\mathrm{index}} = \mathcal{S'}$
			\State \qquad $\mathbf{val}_{\mathrm{index}} = \mathrm{tmp}$
			\State \textbf{end for}
			\State Find the maximum objective value $\mathbf{val}_i$ of $\{\mathbf{val}_1, \mathbf{val}_2, \ldots, \mathbf{val}_G \}$
			\State Find the user subset $\mathcal{S}_i$ corresponding to $\mathbf{val}_i$
			\State Final selected user subset $\mathcal{S} = \mathcal{S}_i$
		\end{algorithmic}
		\label{Scheduling}
	\end{algorithm}

	\begin{figure}
		\centering\includegraphics[scale=0.27]{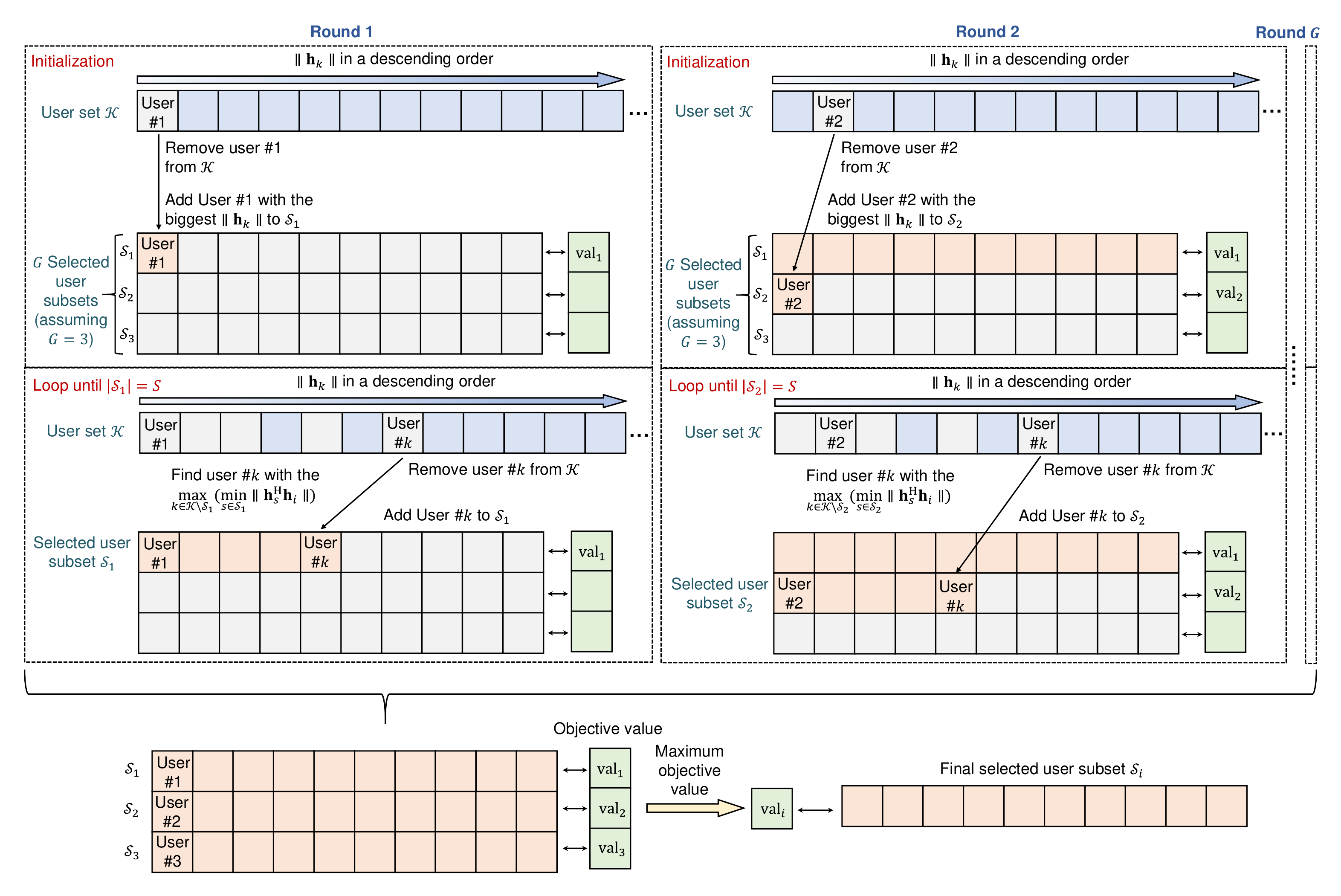}
		\caption{An illustration of greedy low-complexity user scheduling policy.}
		\label{figure_policy_greedy}
	\end{figure}
	
	The complexity of the proposed ``Policy\_greedy'' user scheduling policy is $\mathcal{O}\left(K^{2}\right)$.
	The proposed scheduling method, however also sub-optimal, is more computationally efficient than the iterative method proposed in Section III whose complexity is $\mathcal{O}\left(N_r^{6}\log_{2}\frac{1}{\delta}\right)$.
	After the selected user subset $\mathcal{S}$ is obtained, the aggregate beam steering vector $\mathbf{m}$ can be optimized using Algorithm \ref{SubA}.

\section{Experimental Evaluation and Analysis}
	
This section presents experimental results to evaluate the performance of the proposed iterative aggregate beam steering and user scheduling method (referred to as ``Iterative''), the channel-based low-complexity policy (``Policy''), and the greedy policy (``Policy\_greedy'') in terms of MSE and learning performance.
In our experimental settings, the total number of users $K$ ranges from 20 to 200, from which we select a subset $\mathcal{S}$ to participate in the FL learning. If not otherwise specified, the size of $\mathcal{S}$ defaults to $S=10$.
Simulations were conducted under varying numbers of antennas $N_r = \{8, 16\}$ at the aggregator. The maximum transmit power defaults to $P = 0\;{\rm{dBm}}$. To showcase the advantages of the proposed methods, we compare them against several reference approaches:
	\begin{enumerate}
		\item\textit{Subgradient:} This approach jointly optimizes the beamforming vector and user subset selection using the projected subgradient method \cite{kim2023}. As a low-complexity solution, it is in orders of magnitude faster than conventional SDR-based algorithms.
\item \textit{Iterative reweighted beamforming:}
		This heuristic method is extracted from \cite{li2024antenna} and denoted as ``Iterative reweighted''. At each iteration, user channels are weighted according to their current projection strength, encouraging the beamformer to favor directions aligned with stronger users.
\item \textit{Reinforcement learning:} Although \textbf{Reinforcement Learning (RL)} has not to be employed for the joint beam steering and user scheduling problem, its effectiveness in similar decision optimization tasks motivated us to build an Actor-Critic network based on \cite{zhang2025joint}. We altered the action space to encompass continuous beamforming control and discrete user selection, and conducted simulations to for comparison.
		\item \textit{Random aggregate beamforming and device selection design:} Its aggregate beamforming vector is randomly sampled from a complex unit sphere $\mathbb{C}^N$, and the device subset is determined by sorting the equivalent channel power $\Vert \mathbf{m}^\mathsf{H}\mathbf{h}_k \Vert^2$ \cite{Xu2024Random}. This method is labeled as ``Random beamforming'' in the comparisons.
		\item \textit{Random device selection and beamforming design:} This method randomly selects $S$ users out of $K$ users, followed by optimizing the aggregate beamforming vector $\mathbf{m}$ using DC technique. It is labeled as ``Random selection''.
	\end{enumerate}
	
\subsection{MSE Performance}
\label{MSE_performance}

Overall, the normalized $\mathrm{MSE/\sigma^2}$ decreases as the number of total users $K$ increases as depicted in Fig.~\ref{MSE-K_Nr}. The only exception is the ``Random selection'' method, which fluctuates at a consistently higher value. This demonstrates that the proposed methods effectively harness user diversity, whereas ``Random selection'' fails to do so.
When comparing the two subfigures in Fig.~\ref{MSE-K_Nr}, it becomes clear that $\mathrm{MSE/\sigma^2}$ values across all methods decrease as $N_r$ increases. This trend can be attributed to the additional degrees of freedom provided by a higher number of antennas at the aggregator, which enables better alignment of deviated signals.

Given the same number of $N_r$, ``Policy\_greedy'' achieves the best MSE performance, with ``Iterative'' slightly outperforming ``Policy'' among three proposed methods.
Additionally, although ``Subgradient'' achieves lower $\mathrm{MSE/\sigma^2}$ for small $K$, its performance deteriorates obviously as $K$ increases. This degradation stems from the rapidly expanding complexity of the solution space as $K$ increases, which complicates the optimization process. On the other hand, ``Policy\_greedy'' maintains a consistent advantage by effectively sampling and leveraging user diversity.
The ``Iterative reweighted'' method, while computationally efficient, is an approximate technique that incrementally approaches the optimum via coarse adjustments, thus demonstrating inferior performance.
The poorer performance of the RL approach is attributed to the large combined action space of beam steering and user selection, which dilutes gradient signals and hinders convergence. Hence, a slight rise in $\mathrm{MSE/\sigma^2}$ can be observed when $K$ becomes excessively large.

\begin{figure*}
	\centering \subfigure[]{\includegraphics[scale=0.464]{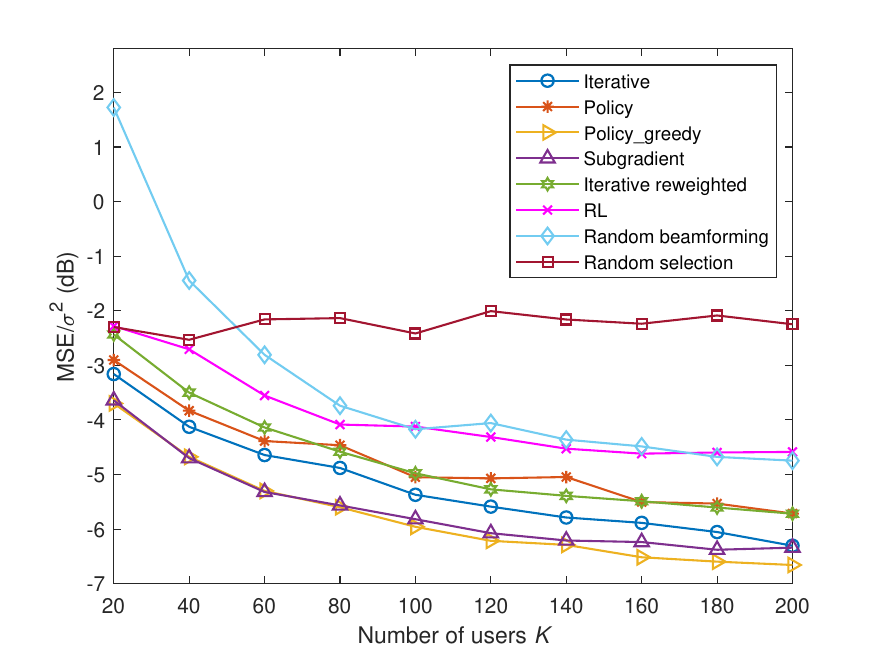}} \subfigure[]{\includegraphics[scale=0.463]{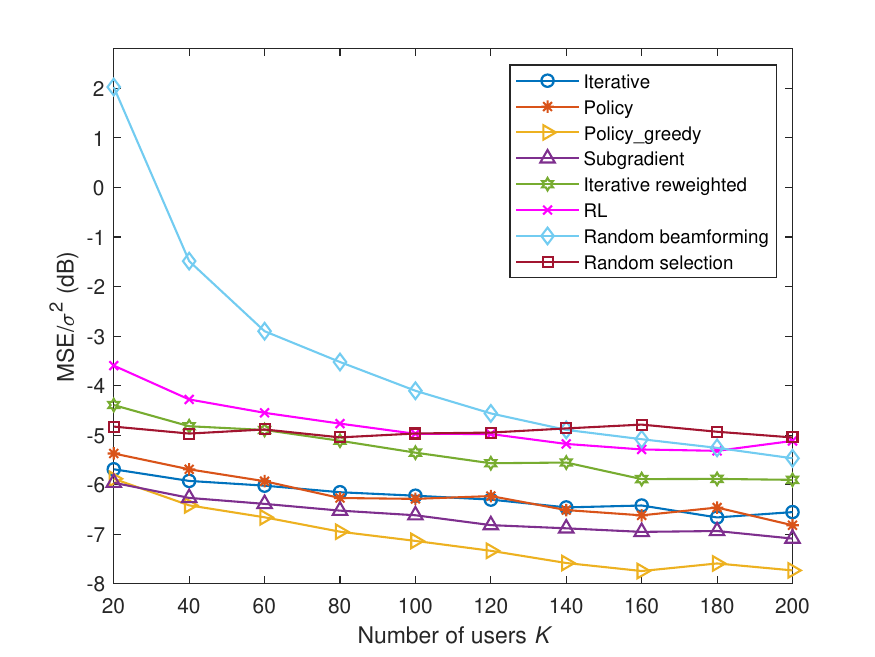} \label{MSE-K_Nr=16}}
	\caption{Aggregate error versus the number of users $K$. (a) $N_{r}=8$. (b) $N_{r}=16$.}
	\label{MSE-K_Nr}
\end{figure*}

\subsection{Cumulative Distribution Function of MSE}

To further examine the distribution of $\mathrm{MSE/\sigma^2}$, we plotted its cumulative distribution function (CDF) under four experimental settings where $N_r = 8$ and $K = [80, 120, 160, 200]$. For each setting, $50$ independent simulations were conducted. The results are presented in Fig.~\ref{CDF}.

Overall, the $\mathrm{MSE/\sigma^2}$ values achieved by the proposed methods are significantly lower than those of the baseline approaches. When $K=80$, the CDF curves of ``Subgradient'' and ``Policy\_greedy'' intersect, suggesting no clear advantage for either method at this stage. However, as $K$ increases, the CDF curve of ``Policy\_greedy'' progressively surpasses that of ``Subgradient'', reflecting its superior performance in handling larger user sets. Additionally, as $K$ grows, the mean and variance gradually decline for all methods except ``Random selection''. These observations align with the analysis in Section \ref{MSE_performance}, which provides further evidence of the robustness and effectiveness of the proposed methods.
	
\begin{figure*}
	\centering
	\subfigure[]{\includegraphics[scale = 0.464]{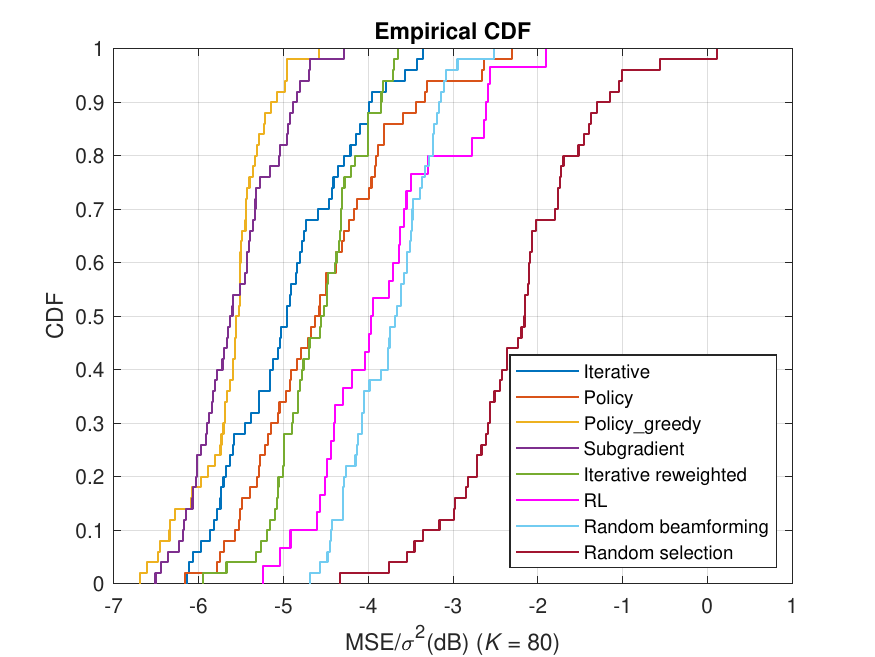}}
	\subfigure[]{\includegraphics[scale = 0.464]{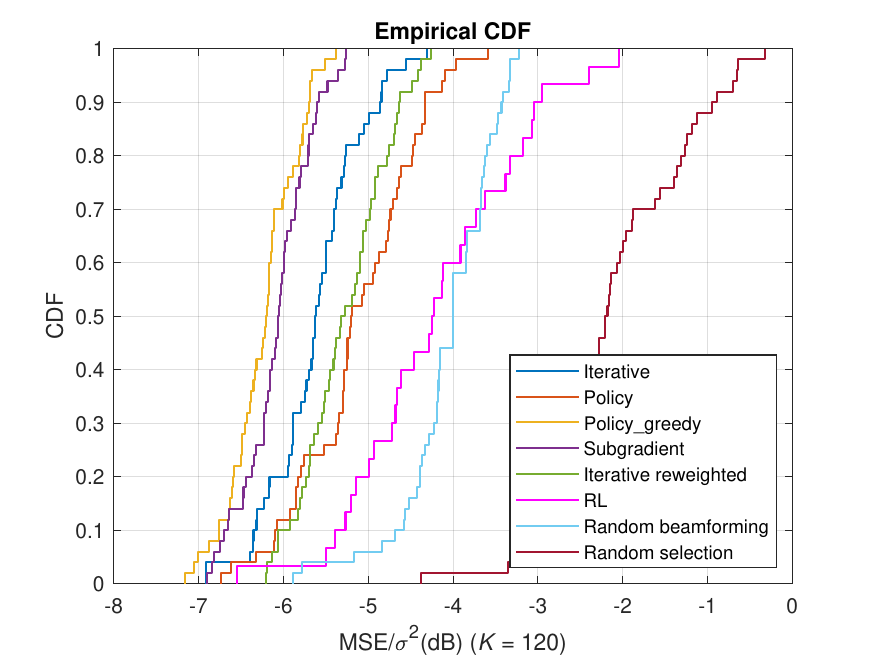}}
	\subfigure[]{\includegraphics[scale = 0.464]{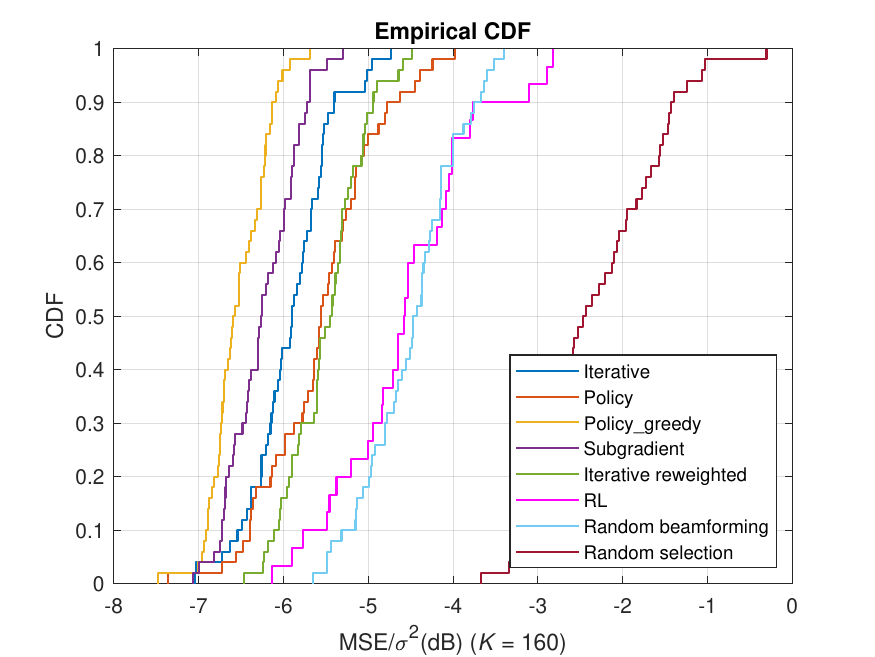}}
	\subfigure[]{\includegraphics[scale = 0.464]{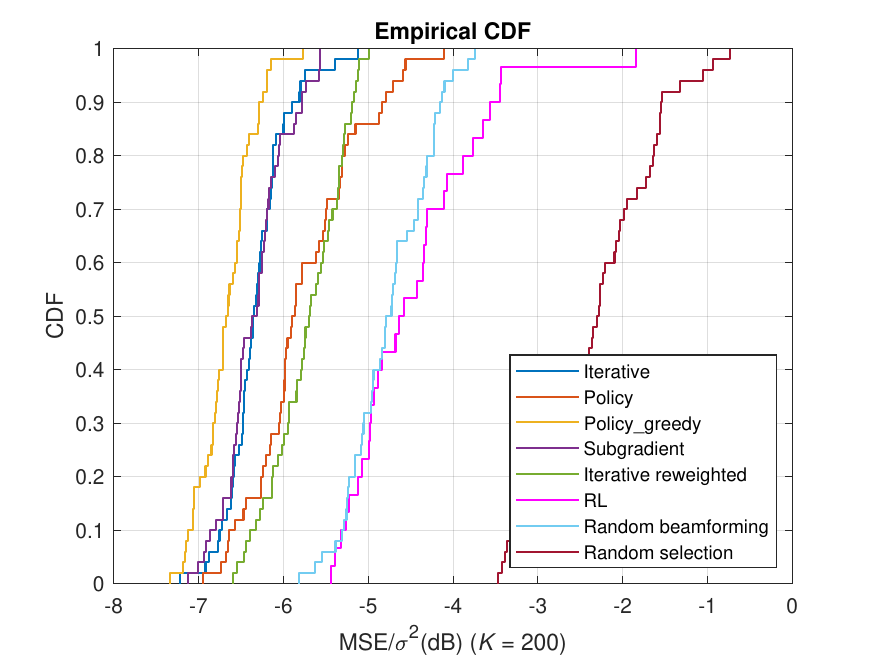}}
	\caption{CDF of $\mathrm{MSE/\sigma^2}$. (a) $K = 80$. (b) $K = 120$. (c) $K = 160$. (d) $K = 200$.}
	\label{CDF}
\end{figure*}
	
\subsection{Efficiency Performance}

In this subsection, we additionally provide the theoretical computational complexity and practical runtime for the proposed methods and baselines in Table. \ref{tab:efficiency}, where $T$ and $d$ denote the iterative epochs and embedding dimension.
Among three proposed methods, the ``Iterative'' method has notably higher computational complexity and execution time due to repeated iterations for solving $\mathbf{m}$. Meanwhile, the ``Policy\_greedy'' and ``Policy'' methods perform a single computation of $\mathbf{m}$, greatly reducing the overall computational burden. Since the ``Policy\_greedy'' selects multiple candidate user subsets, it consumes slightly more runtime than the ``Policy''. However, this additional time cost is minimal in large-scale user scenarios and is well justified by the significant improvement in MSE.
Besides, despite their reduced computational time, the ``Subgradient'' and ``Iterative reweighted'' and ``Random beamforming'' compromise on the MSE performance.
Therefore, ``Policy\_greedy'' provides superior MSE performance with reasonable computational overhead, thereby striking a favorable balance between MSE and efficiency performance.

Furthermore, we provide a theoretical comparison in terms of convergence rate.
In the joint beam steering and user scheduling context, the ``Subgradient'' method converges at a rate of $\mathcal{O}(1/\sqrt{T})$ \cite{subgradient2019}, and the ``Iterative reweighted'' method exhibits sublinear convergence rate \cite{daubechies2010iteratively}.
Although RL lacks theoretical convergence guarantees, it generally improves performance incrementally through iterations. However, when the action space is extremely large, such as the problem studied in this paper, the convergence speed is typically slow.
Besides, the convergence of two random methods is not guaranteed.
In comparison, according to Theorem \ref{the:bound}, the proposed DC algorithm achieves linear convergence rate, which implies that fewer iterations are needed to approach the optimum.

\begin{table}
	\centering
	\caption{Computational Complexity and Practical Runtime.}
	\resizebox{1.0\columnwidth}{!}{
		\begin{tabular}{c|cc||c|cc}
			\toprule
			Methods & Complexity & Runtime (s) & Methods & Complexity & Runtime (s) \\
			\midrule
			Iterative & $\mathcal{O}\left(T\left(N_r^{6}\log_{2}\frac{1}{\delta} + K\log(K)\right)\right)$ & 17.1688 & Iterative reweighted & $\mathcal{O}(N_rK)$ & 0.0027 \\
			Policy & $\mathcal{O}\left(N_r^{6}\log_{2}\frac{1}{\delta} + K^2\right)$ & 3.1786 & RL & $\mathcal{O}(T(N_rKd+SKd))$ & 840.123\\
			Policy\_greedy & $\mathcal{O}\left(N_r^{6}\log_{2}\frac{1}{\delta} + K^2\right)$ & 3.4501 & Random beamforming & $\mathcal{O}(N_r +K\log(K))$ & 0.0008 \\
			Subgradient & $\mathcal{O}(N_rK)$ & 0.0052 & Random selection & $\mathcal{O}\left(N_r^{6}\log_{2}\frac{1}{\delta}+K\right)$ & 12.4545 \\
			\bottomrule
	\end{tabular}}
	\label{tab:efficiency}
\end{table}
	
\subsection{Influence of $S$ on MSE Performance}

\begin{figure*}
	\centering
	\subfigure[]{\includegraphics[scale = 0.57]{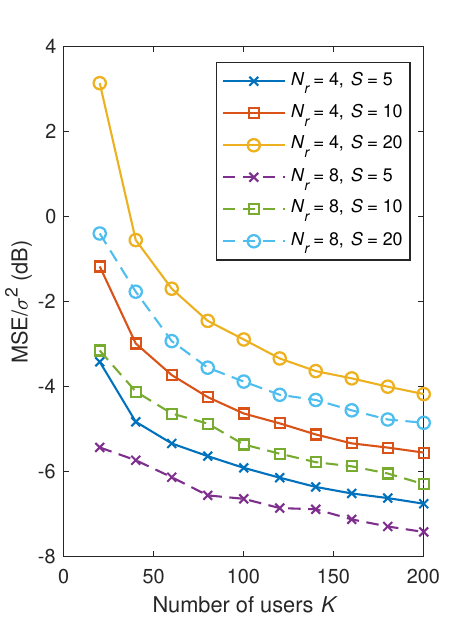}}
	\subfigure[]{\includegraphics[scale = 0.57]{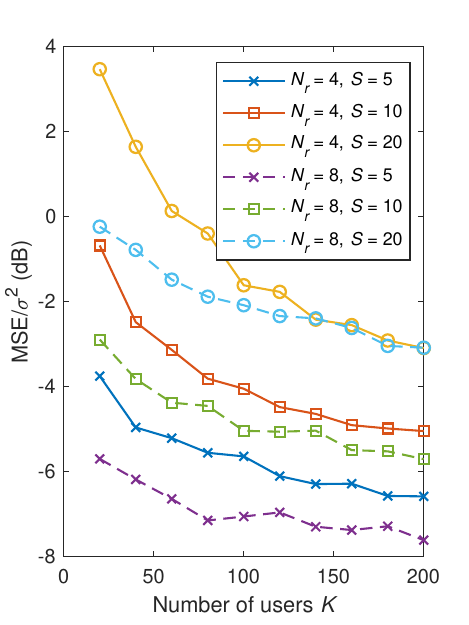}}
	\subfigure[]{\includegraphics[scale = 0.57]{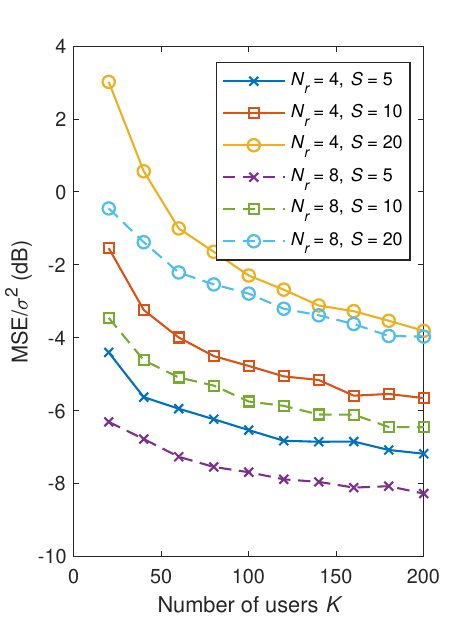}}
	\caption{Aggregate error versus the number of users $K$. (a) Iterative. (b) Policy. (c) Policy\_greedy.}
	\label{MSE-K_S}
\end{figure*}

In this subsection, we examine how varying the user subset size  $S$ affects MSE performance. Experiments were conducted with $S=[5, 10, 20]$ under two conditions $N_r=4$ and $N_r=8$. The results for the ``Iterative'', ``Policy'' and ``Policy\_greedy'' ($G=5$) methods are illustrated in Fig. \ref{MSE-K_S}.

These results show that, for a given $N_r$, increasing $S$ leads to higher $\mathrm{MSE/\sigma^2}$. This outcome highlights that involving more users in model updates introduces greater aggregation errors. Conversely, for a fixed $S$, increasing $N_r$ results in lower $\mathrm{MSE/\sigma^2}$, underscoring the advantage of equipping the aggregator with more antennas to enhance signal alignment. Note that, these results are also consistent with the trends observed in Section \ref{MSE_performance}.
	
\subsection{Influence of $G$ on MSE Performance}

\begin{figure*}
	\centering
	\subfigure[] {\includegraphics[scale=0.464]{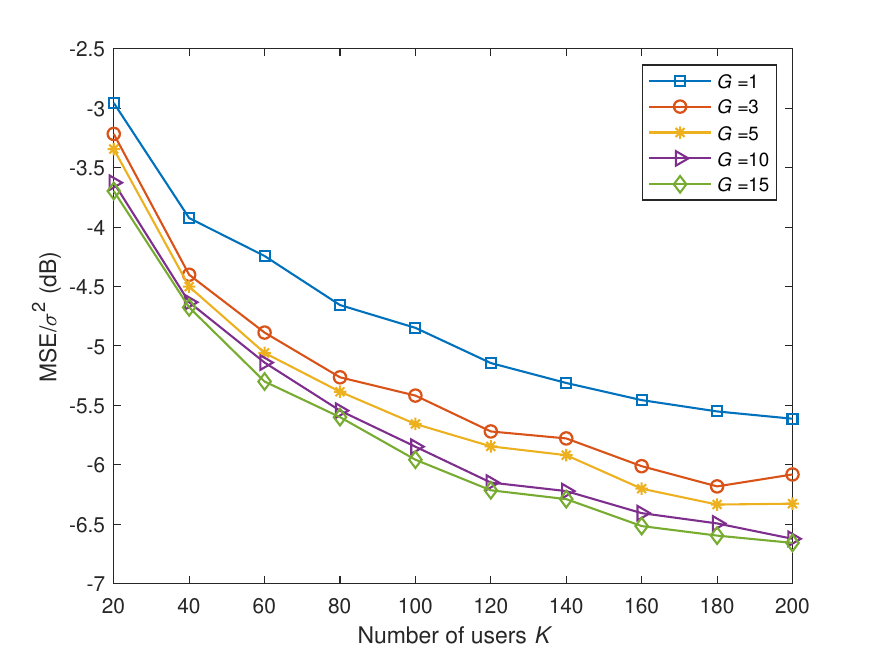}\label{influence_G_S10}}
	\subfigure[] {\includegraphics[scale=0.464]{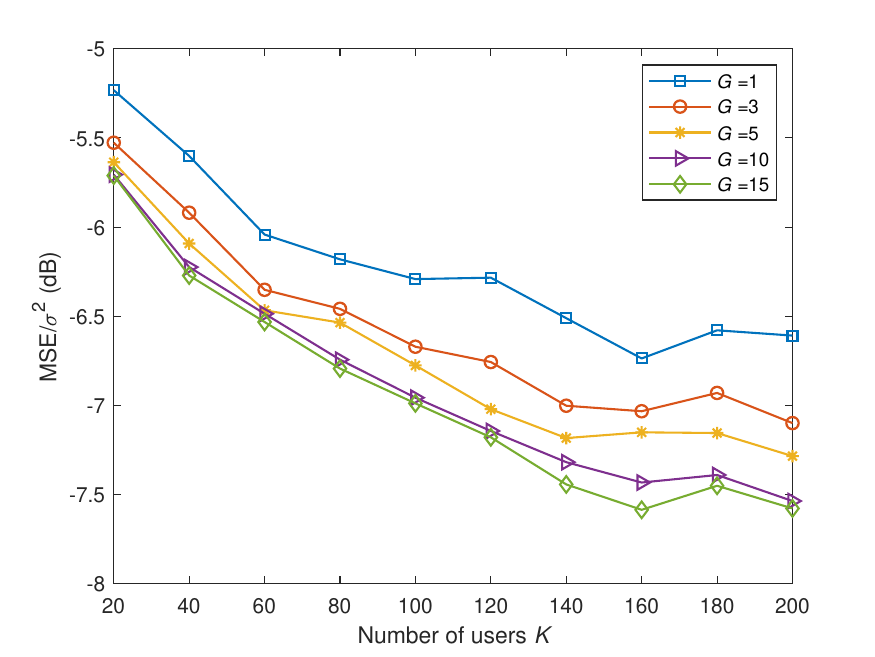}\label{influence_G_S20}}
	\caption{Aggregate error versus the number of users $K$. (a) $N_r=8$. (b) $N_r=16$. }
	\label{influence_G}
\end{figure*}

To evaluate the impact of the greedy number $G$ on the MSE performance, we conducted experiments with $G \in [1, 3, 5, 10, 15]$ for the ``Policy\_greedy'' method. It is noteworthy that the ``Policy'' method is equivalent to ``Policy\_greedy'' when $G=1$. The outcomes for $N_r=8$ and $N_r=16$ are presented in Fig. \ref{influence_G_S10} and Fig. \ref{influence_G_S20}, respectively.

The results in Fig.~\ref{influence_G} indicate that increasing $G$ reduces $\mathrm{MSE/\sigma^2}$, particularly when $G$ is small. However, as $G$ becomes sufficiently large, the marginal benefit diminishes. This trend suggests that selecting an appropriate $G$ value is crucial to balancing improved MSE performance against increased computational cost. The findings highlight that while larger $G$ values offer gains in accuracy, they also come with greater computational expenses. Therefore, a carefully chosen $G$ can optimize both performance and efficiency.
	
\subsection{Learning Performance}

\begin{figure*}
	\centering
	\subfigure[] {\includegraphics[scale=0.464]{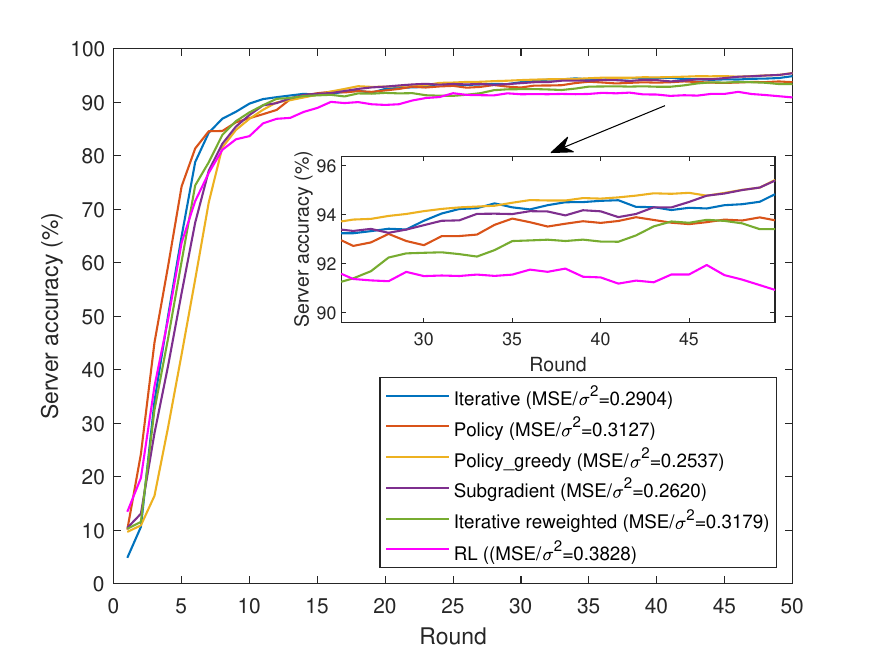}\label{iid}}
	\subfigure[] {\includegraphics[scale=0.464]{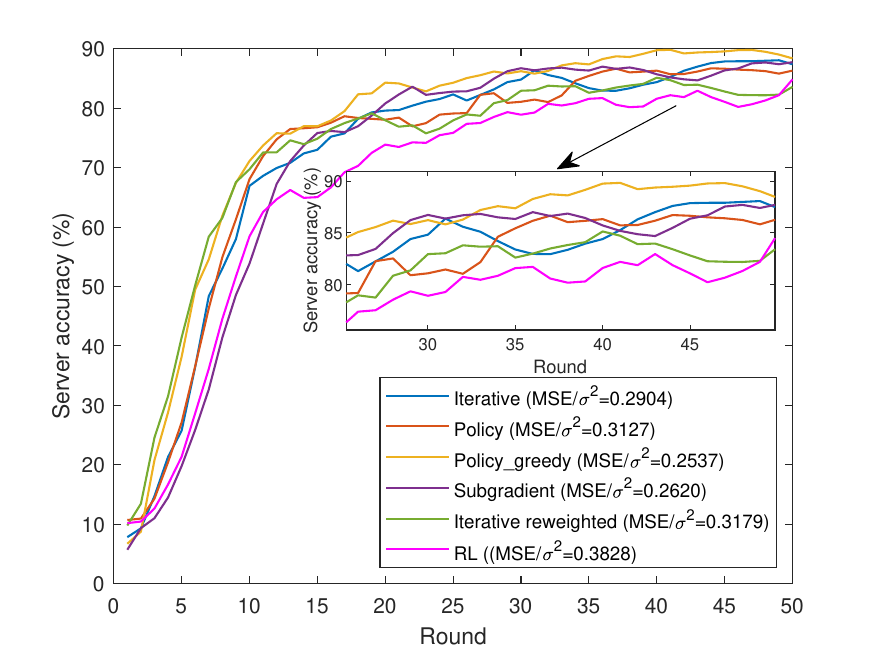}\label{niid}}
	\caption{Test accuracy. (a) i.i.d MNIST. (b) Non-i.i.d MNIST.}
	\label{learning_performance}
\end{figure*}
	
In this final subsection, we investigate the influence of MSE on learning performance. We applied the proposed methods to an FL framework and we trained a standard image classification task on the well-known MNIST dataset. MNIST, a widely recognized benchmark for image classification, consists of $70,000$ black-and-white handwritten digit images, with $60,000$ images allocated for training and $10,000$ for testing. For this task, we utilized the LeNet5 network architecture within the FedAvg framework. The mini-batch technique was employed with a batch size of $64$, while the learning rates were set to $0.1$ for each user and $0.5$ for the aggregator.

Simulations were carried out under a scenario where $N_r = 8$ were equipped at the aggregator, and $S = 10$ users were selected from a pool of $K = 100$ users. The training process was evaluated under both independent and identically distributed (i.i.d.) and non-independent and identically distributed (non-i.i.d.) data conditions among users. The non-i.i.d. scenario represents a more realistic setting, as local training data often depends on specific user behaviors, introducing various skews such as attribute skew, label skew, and temporal skew \cite{Zhao2022Non-iid}. Furthermore, theoretical and experimental analyses have shown that non-i.i.d. data significantly reduces the test accuracy of FL systems \cite{Gao2022Feddc}. Consequently, it is essential to validate that the proposed methods remain effective under non-i.i.d. conditions.

The simulation results reveal that the two random baseline methods performed significantly worse than the other approaches; thus, only the remaining seven methods were included in the analysis, which are ``Iterative'', ``Policy'', ``Policy\_greedy'', ``Subgradient'',  ``Iterative reweighted'', and ``RL''. Their average $\mathrm{MSE/\sigma^2}$ values are  $0.2904$, $0.3127$, $0.2537$, $0.2620$, $0.2620$, $0.3179$ and $0.3828$, respectively.
The corresponding server accuracy results for both i.i.d. and non-i.i.d. scenarios are presented in Fig.~\ref{learning_performance}\subref{iid} and Fig.~\ref{learning_performance}\subref{niid}. Under non-i.i.d. conditions, the fluctuations in server accuracy are more pronounced compared to the i.i.d. case, and the overall recognition accuracy is notably lower. Given the increased challenges posed by data heterogeneity, this observation aligns with our expectation. Additionally, the results demonstrate a clear relationship between $\mathrm{MSE/\sigma^2}$ and recognition accuracy: lower $\mathrm{MSE/\sigma^2}$ values correspond to higher recognition accuracy due to reduced transmission distortion. Among the evaluated methods, ``Policy\_greedy'' consistently achieved the best performance, which further highlights its robustness and effectiveness in both i.i.d. and non-i.i.d. scenarios.

\section{Conclusion}

In this article, we explore the deployment of FL within a large-scale wireless network to support diversified mobile IoT applications across a multitude of edge users.
While the adoption of AirComp has been identified as an approach to improve the communication efficiency, it inevitably produces aggregate error due to the fading and noisy channels.
Thus, we propose an iterative user scheduling and aggregate beam steering design to align the deviated signal.
Furthermore, to reduce the computational demands associated with the iterative process, we develop a low-complexity user scheduling policy.
Stemming from the analysis of channel characteristics, the user subset can be directed determined without iteration.
A greedy user scheduling policy is further proposed to approximate the optimal user subset.
Numerical simulations have confirmed that our proposed methods are superior to the benchmark in terms of both the MSE performance and the prediction accuracy on MNIST10 dataset.
	
\begin{acks}
This work was supported in part by the National Natural Science Foundation of China under Grant No.~62001103 and the Fundamental Research Funds for the Central Universities under Grant No.~2242023R40005.
\end{acks}
	
\bibliographystyle{ACM-Reference-Format}
\bibliography{reference}

\appendix
\section{NP-Hardness Proof of Problem (\ref{eq:SIMO2})}
\label{hardness}

We establish the NP-hardness proof of Problem \eqref{eq:SIMO2} via a reduction from the maximum clique problem, which is known to be NP-hard \cite{gutin2004independent}.
Firstly, we present the definition of the maximum clique problem.
\begin{definition}
	Maximum Clique Problem:
	Given an undirected graph $\mathcal{G}=(\mathcal{V}, \mathcal{E})$, where $\mathcal{V}$ and $\mathcal{E}$ represent the set of vertices and edges. Find the largest subset of vertices $\mathcal{S}\subseteq \mathcal{V}$ such that every pair of vertices in $\mathcal{S}$ is connected by an edge, (i.e., $\mathcal{S}$ forms a clique).
\end{definition}

Then, we construct a polynomial-time reduction from Problem \eqref{eq:SIMO2} to the maximum clique problem.
\begin{proof}
We first project the users and their mutual relationships into graph $\mathcal{G}$.
Let $\mathcal{K}= \left\{\mathbf{h}_1, \mathbf{h}_2, \ldots, \mathbf{h}_K\right\}$ denotes the set of vertices of $\mathcal{G}$.
The edges $(v_i, v_j) \in \mathcal{E}$ can be defined based on the inner product of $\mathbf{h}_i$ and $\mathbf{h}_j$. Specifically, there is an edge between $v_i$ and $v_j$ if and only if $\Vert\mathbf{h}_i^\mathsf{H}\mathbf{h}_j\Vert$ exceeds a given threshold $\tau >0$.

According to Theorem \ref{thmtheorem1} and Theorem \ref{thmtheorem2}, the problem \eqref{eq:SIMO2} can be transformed to select a user subset $\mathcal{S}$ from $\mathcal{K}$ to
maximize the minimum projection value $\Vert \mathbf{h}_i^\mathsf{H}\mathbf{h}_j\Vert$ across the selected subset, i.e., $i, j \in \mathcal{S}$.
This optimization problem is equivalent to select the maximum clique $\mathcal{S}$ from $\mathcal{K}$. This guarantees that the magnitude of inner product of any pair of users exceeds $\tau$.
However, the constraint $\vert\mathcal{S}\vert = S$ is not yet satisfied. Hence, the threshold $\tau$ is adjusted through an iterative update process.
Specifically, if the number of vertices in the clique exceeds $S$, then increase the threshold $\tau$; otherwise, $\tau$ is reduced until the selected clique $\mathcal{S}$ satisfies $\vert\mathcal{S}\vert=S$.
In summary, since finding the maximum clique in each iteration is already an NP-hard problem, the overall problem is at least as hard, if not harder.
\end{proof}

\section{Proof of Theorem \ref{the:bound}}
\label{proof_bound}

\begin{proof}
For an arbitrary user $k$, we first define $\mathbf{H}_k = \mathbf{h}_k\mathbf{h}_k^\mathsf{H}$. For notational simplicity, the subscript $k$ is omitted and $\mathbf{H}_k$ is written as $\mathbf{H}$. If $f(\mathbf{m})=\mathbf{m}^\mathsf{H}\mathbf{H}\mathbf{m}$ satisfies the PL inequality, there exists a constant $\eta>0$ such that the following inequality holds
	\begin{equation}
		\vert f(\mathbf{m}) - f(\bar{\mathbf{m}})\vert \leq \frac{1}{2\eta}\Vert \xi \Vert^2 , \forall \xi \in \operatorname{co}(\nabla f(\mathbf{m})),
	\end{equation}
	where $\bar{\mathbf{m}}$ is a critical point, that is, $\nabla f(\bar{\mathbf{m}}) = \mathbf{H}\bar{\mathbf{m}}=0$. Besides,  $\operatorname{co}(\nabla f(\mathbf{m}))$ is the convex hull of $\nabla f(\mathbf{m})$.
	Due to the smoothness of $f(\mathbf{m})$, it can be represented as:
	\begin{equation}
		\vert f(\mathbf{m}) - f(\bar{\mathbf{m}})\vert \leq \frac{1}{2\eta}\Vert \mathbf{H}\mathbf{m} \Vert^2 .
	\end{equation}
	This equation imposes an upper bound on the gap between the derived suboptimal solution and the optimal solution.
	
	Then, the convergence rate is derived through performance estimation \cite{abbaszadehpeivasti2024rate}.
	If $f_1$ and $f_2$ are strictly differentiable, we have $\operatorname{co}(\nabla f(\mathbf{m}))=\nabla f_1 - \nabla f_2$.
	Consequently, denote $f^n$ is the $n$-th iteration value of $f$ and function $g=\nabla f$, the performance can be estimated with the PL inequality as
	\begin{equation}
	\begin{aligned}
		&\max \; \frac{(f_1^2 - f_2^2) - f^*}{(f_1^1 - f_2^1) - f^*} \\
		&\text{s.t.} \quad
		\frac{1}{2(1-\frac{\mu_1}{L_1})} \left( \frac{1}{L_1} \left\| g_1^i - g_1^j \right\|^2 + \mu_1 \left\| x^i - x^j \right\|^2
		- \frac{2\mu_1}{L_1} \left\langle g_1^j - g_1^i, x^j - x^i \right\rangle \right)\\
		&\qquad 
\leq f_1^i - f_1^j - \left\langle g_1^j, x^i - x^j \right\rangle \quad i, j \in \{1,2\} \\
		&\qquad \frac{1}{2(1-\frac{\mu_2}{L_2})} \left( \frac{1}{L_2} \left\| g_2^i - g_2^j \right\|^2 + \mu_2 \left\| x^i - x^j \right\|^2
		- \frac{2\mu_2}{L_2} \left\langle g_2^j - g_2^i, x^j - x^i \right\rangle \right)\\
		&\qquad \leq f_2^i - f_2^j - \left\langle g_2^j, x^i - x^j \right\rangle \quad i, j \in \{1, 2\} \\
		&\qquad f_1^k - f_2^k \geq f^* \quad k \in \{1, 2\} \\
		&\qquad g_2^1 = g_1^2 \\
		&\qquad \left( f_1^k - f_2^k \right) - f^* \leq \frac{1}{2\eta} \left\| g_1^k - g_2^k \right\|^2, \quad k \in \{1, 2\},
	\end{aligned}
	\label{eq:rate1}
	\end{equation}
	where $ f^*=f(\bar{\mathbf{m}})$, $\mu$ represents the strong convexity parameter and $L$ denotes the smoothness parameter.
	By aggregating the constraints in \eqref{eq:rate1}, we can prove that the DC algorithm convergence at a linear rate.
	Without loss of generality, we assume that $\mu_1=\mu_2=0$ and $f^*=0$. We can obtain that
	\begin{equation}
		\label{eq:rate2}
		\begin{aligned}
			& \left(f_1^2-f_2^2\right)-f^*-\left(\frac{1-\frac{\eta}{L_1}}{1+\frac{\eta}{L_2}}\right)\left(\left(f_1^1-f_2^1\right)-f^*\right)+\left(\frac{1}{1+\frac{\eta}{L_2}}\right) \\
			& \times\left(f_1^1-f_1^2-\left\langle g_1^2,x^1-x^2\right\rangle-\frac{1}{2L_1}\left\Vert g_1^1-g_1^2\right\Vert^2\right) \\
			& +\left(\frac{1}{1+\frac{\eta}{L_2}}\right)\left(f_2^2-f_2^1-\langle g_1^2,x^2-x^1\rangle-\frac{1}{2L_2}\left\Vert g_1^2-g_2^2\right\Vert^2\right)+\left(\frac{\frac{\eta}{L_1}}{1+\frac{\eta}{L_2}}\right) \\
			& \times\left(\frac{1}{2\eta}\left\|g_1^1-g_1^2\right\|^2-f_1^1+f_2^1\right)+\left(\frac{\frac{\eta}{L_2}}{1+\frac{\eta}{L_2}}\right)\left(\frac{1}{2\eta}\left\|g_1^2-g_2^2\right\|^2-f_1^2+f_2^2\right)=0.
	\end{aligned}\end{equation}
	Since all the multipliers in \eqref{eq:rate2} are non-negative, for any feasible solution, it holds that
	\begin{equation}
		f(x^2)-f^*-\left(\frac{1-\frac{\eta}{L_1}}{1+\frac{\eta}{L_2}}\right)\left(f(x^1)-f^*\right)\leq0,
	\end{equation}
	which establishes the linear convergence of the DC method.
	\end{proof}

\end{document}